\documentclass[envcountsame]{article}

\usepackage[T1]{fontenc}
\usepackage{a4wide}

\usepackage{adjustbox}
\usepackage{algorithm}
\usepackage{algorithmicx}
\usepackage[noend]{algpseudocode}
\usepackage{amsfonts,amsmath,amssymb,amsthm} 
\usepackage[title]{appendix}
\usepackage[english]{babel}
\usepackage{booktabs}
\usepackage{braket}
\usepackage[labelfont=bf]{caption}
\usepackage{colonequals}
\usepackage[shortlabels]{enumitem}
\usepackage{epigraph}
\usepackage{forest}
\usepackage{graphicx}
\usepackage{listings}
\usepackage{manyfoot}
\usepackage{mathrsfs}
\usepackage{mathtools}
\usepackage{multirow}
\usepackage{pgf}
\usepackage{pstricks}
\usepackage{setspace}
\usepackage{stackengine}
\usepackage{stix}
\usepackage{stmaryrd}
\usepackage{subcaption}
\usepackage{tabularx}
\usepackage{textcomp}
\usepackage{thmtools,thm-restate}
\usepackage{verbatim}

\usepackage{tikz,tikz-cd}
\usetikzlibrary{arrows,arrows.meta,automata,backgrounds,calc,chains,decorations.pathmorphing,decorations.pathreplacing,decorations.text,matrix,petri,positioning,shapes,shapes.geometric,shapes.multipart,shapes.symbols,trees}
\tikzstyle{nod}= [circle, draw,inner sep=0pt, minimum size=0.5cm]
\tikzset{
modal/.style={>=stealth’,shorten >=1pt,shorten <=1pt,auto,node distance=1.5cm,
semithick},
world/.style={circle,draw,minimum size=0.5cm,fill=gray!15},
point/.style={circle,draw,inner sep=0.5mm,fill=black},
reflexive above/.style={->,loop,looseness=7,in=120,out=60},
reflexive below/.style={->,loop,looseness=7,in=240,out=300},
reflexive left/.style={->,loop,looseness=7,in=150,out=210},
reflexive right/.style={->,loop,looseness=7,in=30,out=330}
}

\usepackage{hyperref}
\usepackage[capitalize]{cleveref}

\newcommand{\agents}{\mathcal{A}}
\newcommand{\calM}{\mathcal{M}}
\newcommand{\calN}{\mathcal{N}}

\newcommand{\prop}{\mathit{Prop}}
\newcommand{\atoms}{\prop}

\newcommand{\APB}{\mathit{A\!P\!B}}

\newcommand{\C}{\mathord{\mapsto}}
\newcommand{\U}{{\sf U}}

\newcommand{\ce}{\colonequals}

\newtheorem{theorem}{Theorem}
\newtheorem{corollary}[theorem]{Corollary}%
\newtheorem{proposition}[theorem]{Proposition}%

\theoremstyle{remark}%
\newtheorem{example}[theorem]{Example}%
\newtheorem{remark}[theorem]{Remark}%

\theoremstyle{definition}%
\newtheorem{definition}[theorem]{Definition}%

\usepackage[T1]{fontenc}
\usepackage[utf8]{inputenc}
\usepackage{authblk}
\renewcommand{\phi}{\varphi}
\renewcommand{\models}{\vDash}
\newcommand{\notmodels}{\nvDash}
\newcommand{\upd}{\mathrel{\!\mid\!}}
\newcommand{\updpriv}[1]{\mathrel{\!\mid_{#1}\!}}
\newcommand{\ASF}{\mathit{A\!S\!F}}

\usepackage{stackengine}

\title{A~priori Belief Updates\\ as a Method for Agent Self-Recovery
\thanks{This research was funded by the Austrian Science Fund (FWF) project ByzDEL (P33600).}
}
\author[ ]{Giorgio Cignarale}
\author[ ]{ Roman Kuznets}
\affil[ ]{TU Wien, Austria}
\affil[ ]{\textit {\{giorgio.cignarale,roman.kuznets\}@tuwien.ac.at}}
\date{}
\providecommand{\keywords}[1]{\textbf{\textit{Keywords:}} #1}

\begin{document}

\maketitle

\begin{abstract}
Standard epistemic logic is concerned with describing  agents' epistemic attitudes  given the current set of alternatives the agents consider possible. 
While distributed systems can (and often are)  discussed without mentioning epistemics, it has been well established that epistemic phenomena lie at the heart of what agents, or processes, can and cannot do. Dynamic epistemic logic (DEL) aims to describe how epistemic attitudes of the agents/processes change based on the new information they receive, e.g., based on their observations of events and actions in a distributed system.
In a broader philosophical view, this  appeals to an \emph{a~posteriori} kind of reasoning, where agents update the set of alternatives considered possible  based on their ``experiences.'' 

Until recently, there was little incentive to formalize \emph{a~priori} reasoning, which  plays a role in designing and maintaining distributed systems, e.g.,
in determining which states must be considered possible by agents in order to solve the distributed task at hand, and consequently in updating these states when unforeseen situations arise during runtime.
With systems becoming more and more complex and large,  the task of fixing design errors ``on the fly'' is shifted to individual agents, such as in
the increasingly popular self-adaptive and self-organizing (SASO) systems.
Rather than updating agents' \textit{a~posteriori beliefs}, this requires modifying their  \textit{a~priori beliefs} about the system's global design  and parameters. 

The goal of this paper is to 
provide a formalization of such a~priori reasoning by using standard epistemic semantic tools, including Kripke models and DEL-style updates, and provide heuristics that would pave the way to streamlining this inherently non-deterministic and \emph{ad~hoc} process for SASO systems.

\end{abstract}

\section{Introduction}
\label{sec:Intro}
\emph{Epistemic logic}~\cite{Hintikka1962-HINKAB} reasons about knowledge and/or beliefs of agents in a multiagent system. 
\emph{Distributed systems} are a type of multiagent systems, with agents often referred to as \emph{processes}, where these processes must coordinate their actions by  communicating via either message passing or  shared memory   in order to accomplish some task~\cite{Lynch_DA,DS:concept_design}. 
Formal epistemic modeling~\cite{bookof4} has proved to be useful for characterizing the distributed system's evolution over time~\cite{HM90}, for deriving impossibility results~\cite{MT88}, and  for determining what processes can compute based on their local state in a given model~\cite{ben2014beyond,GM18:PODC,CGM14}. The reason why epistemic analysis is always relevant in distributed scenarios was recently formalized as the Knowledge of Preconditions Principle in~\cite{Mos15TARK}, which is so universal that it applies even to fault-tolerant distributed systems~\cite{KPSF19:TARK,KuznetsPSF19FroCoS,fire,SchloeglS23}.
Dynamic epistemic logic (DEL)~\cite{Plaza,GerbrandyG97,Baltag2016,ditmarsch2007dynamic} provides tools for analyzing change in agents' epistemic attitudes in response to new information.

The epistemic analysis of distributed systems~\cite{bookof4} and of epistemic puzzles~\cite{vDitmarschK15} routinely relies on agents' common knowledge of the model~\cite{artemov}. In effect, this is used to model agents' common \emph{a~priori} assumptions and enables agents to reason about (higher-order) reasoning of other agents. Note that this type of knowledge differs from what agents learn through communication, independently of whether that communication is public, as in many epistemic puzzles, or private, as is more common in distributed systems. The information agents learn while playing out a puzzle or during a distributed system run is experience-based, \emph{a~posteriori} knowledge. Accordingly, dynamic epistemic logic implements knowledge change by model modifications that reorganize and shrink the already available possibilities, in contrast to the initial epistemic model creating the common space of these possibilities for the agents based on  the puzzle description or distributed system specification.

Therefore, the system designer's task of creating a distributed system to given specifications can be viewed as creating common a~priori knowledge for the agents. The role of a~priori knowledge in the design cycle of distributed systems is analyzed in~\cite{APK}. As argued there, mistakes in a system design would normally require the system designer to initiate the recovery process that amounts to the a~priori knowledge update, better termed \emph{a~priori belief update} due to the fallibility assumption inherent in the situation when system behavior does not match the desired specifications. The fallibility here applies not only to the agents but also to the system designer who failed to account for some factors and/or behaviors. This picture of a~posteriori experiences  triggering  a~priori belief updates is largely based on 
a new philosophical approach to the a~priori vs.~a~posteriori distinction proposed in~\cite{Tahko2008-TAHAND,Tahko2011-TAHAPA}, where it is termed the \emph{bootstrapping relation}. It is important to note that a~priori knowledge/belief is characterized there  as \emph{modal}, i.e., relating to the set of possible states conceived by  agents, and  \textit{fallible}, in the sense that  that the actual world (or a faithful copy thereof) need not be among this set of possible states.

In the case of traditional distributed systems, aberrant a~posteriori behavior prompts the system designer to trigger a new iteration of the design cycle: in a new design phase, she will re-adapt the a~priori system assumptions so as to match the intended system behavior, redesign the affected parts of the implementation, and finally deploy and restart them, initializing the agents with updated a~priori assumptions. Given the trend towards more and more complex and growing distributed systems, however, discovering and recovering from such design errors is increasingly becoming prohibitively costly: the ability to predict and/or monitor  possible behaviors of such a system decreases exponentially, whereas the redesign costs increase dramatically. 

This trend fueled the development of \emph{self-adaptive and self-organizing systems} (\emph{SASO systems})~\cite{ft_self-*,Knowthyself} that have self-reflection and self-adaption capabilities. SASO systems allow processes to access and operate with their own representation of the system, which in turn enables them to update certain design assumptions on their own. In other words, in addition to a~posteriori belief updates, which can be handled by traditional DEL methods, agents in SASO systems are expected to perform a~priori belief updates with the goal of self-correcting their behavior, in response to situations not envisioned by the system designer.

This paper is devoted to developing an epistemic formalization of self-recovery capabilities for agents by means of  a~priori belief updates, implemented in the form of DEL-inspired updates. We focus on self-recovery from an inconsistent state of beliefs and, following Plaza~\cite{Plaza}, illustrate our methods using (variants of) standard epistemic puzzles such as the consecutive numbers and muddy children puzzle.

Let us first illustrate how faulty a~priori  assumptions can derail the progress, say,  in the consecutive numbers puzzle and how human agents might still be able to find a solution by adjusting their a~priori beliefs.\looseness=-1
\begin{example}[Consecutive numbers]
    \label{ex:consec}
Two agents $a$ and $b$ are privately told two natural numbers. In addition, they are publicly told that the two numbers are consecutive (making it common knowledge).
Suppose that $a$ is told number 1 and $b$ is told number 2. They are allowed to state whether they know the other's number or not, but not allowed to communicate their own number. Ordinarily, this instance of the puzzle is solved by the following dialog:
\begin{itemize}
    \item $a$: I don't know your number.
    \item $b$: I don't know your number either.
    \item $a$: Now I know your number.
    \item $b$: Now I know yours too.
\end{itemize}
Here the first statement by $a$ is uninformative. The first statement by $b$ makes it clear that $b$'s number is not 0, which enables $a$ to conclude that $b$'s number must be 2. Since this determination would not have happened were $a$ to hold number 3, now $b$ can conclude that $a$'s number is 1. The standard epistemic modeling of this example involves a Kripke model that each agent is supposed to build based on the rules of the puzzle in a way that makes this model commonly known to both agents. This commonality is based, in Lewis's telling, on the ``suitable ancillary premises regarding [agents'] rationality, inductive standards, and background information'' (\cite[p. 53]{Lewis69}). While agents in epistemic puzzles are routinely considered to be perfect reasoners, which takes care of rationality and inductive standards, the question of background information is much less clear cut.

For our twist on the original formulation, imagine that, unknown to each other, $a$ and $b$ learned different definitions of natural numbers in school: $a$ starts them from 1, while for $b$ number 0 is also  natural. What Lewis called belief in common background information and we call \emph{a~priori beliefs of the agent} is not shared by the agents preventing them from achieving common knowledge. Since natural numbers are routinely assumed (including in the formulation of the consecutive numbers puzzle) to be a well-defined object, each agent incorrectly believes that the other agent shares their  definition of natural numbers. As is to be expected, the common knowledge of the model and of the situation at hand shatters, leading to the following possible conversation:
\begin{itemize}
    \item $a$: I know your number.
    \item $b$: Wait, what? But that is impossible, unless... Ah, I see. Then I know your number too.
\end{itemize}
Here agent $a$ does not consider $(1,0)$ to be a legitimate pair, hence, (correctly) concludes that the numbers are~$(1,2)$. Agent $b$, on the other hand, expects $a$ to consider $(1,0)$ and, hence, does not understand $a$'s~reasoning.
Indeed, according to $b$, if $a$ had 1, he would have hesitated between $(1,0)$ and $(1,2)$, while if $a$~had~$3$, he would have hesitated between $(3,2)$ and $(3,4)$.  
Thus, $a$'s statement is incompatible with $b$'s~view of the world. In the proposed conversation, $b$ does what is natural for a human reasoner: she questions her a~priori assumptions, conceives that there is an alternative understanding of natural numbers as starting from 1, realizes that this is compatible with $a$'s behavior, and updates her a~priori beliefs (about $a$'s a~priori beliefs). Moreover, after this update, $a$'s claim to knowing $b$'s number is only compatible with $a$ having 1: 
were $a$'s~number~$3$, the hesitation between $(3,2)$ and $(3,4)$ would have persisted. This update of $b$'s a~priori beliefs enables her to both explain the situation and arrive at the correct conclusion.
\end{example}

Standard (dynamic) epistemic reasoning, on the other hand, does not provide an adequate explanation. Epistemically, $b$'s beliefs are supposed to become inconsistent. 
In fault-tolerant systems, this often translates to $b$~considering herself and/or the other agent fully byzantine\footnote{Fully byzantine agents can deviate arbitrarily from their original protocols as well as have false memories and erroneous perceptions. Moreover, their goals might not be known to correct agents. As such, they are attributed potentially inconsistent epistemic attitudes and cannot be trusted by correct agents.} and thus completely unreliable \cite{KPSF19:TARK}, making the puzzle not solvable.

\begin{paragraph}{Contributions} 
To the best of our knowledge, there is currently no epistemic modeling and analysis framework that explicitly considers a~priori beliefs and their dynamics, and the issue of epistemic modeling of self-recovery has not been addressed in the literature.
The goal of the present paper is to model a~priori belief updates epistemically, thus, providing agents with self-correcting capabilities.
While the process remains highly non-deterministic in general, as there are multiple possible ways to resolve  design mistakes, we provide some general guidelines and heuristics to guide possible future implementations of such self-recovery operations in SASO systems.
\end{paragraph}

\begin{paragraph}{Paper organization}
We provide some basic definitions of epistemic logic in~\cref{sec:FP}. In \cref{sec:APB_EL} we discuss the guiding principles behind our epistemic approach to agents' a~priori beliefs, highlighting their private nature (\cref{subsec:priv}), their usefulness in conflict resolutions (\cref{subsec:confl}), and their limitations in the higher-order case (\cref{subsec:hoAPR}). \cref{sec:APBupdate} introduces the novel a~priori belief update mechanism for self-recovery (\cref{subsec:generalmech}), and we show its fruitfulness in variants of popular epistemic puzzles (\cref{subsec:examples}), including a~priori belief updates triggered by public announcements (\cref{subsec:example_PA}), simultaneous (and independent) a~priori updates by several agents (\cref{subsec:example_Simultaneous}), and  a~priori updates (\cref{subsec:example_unsuccess}) that do not achieve the desired goals.
Some heuristics for the update synthesis problem are provided in \cref{subsec:heuristics}. \cref{sec:propertiesAPBU} lists some useful properties of a~priori belief updates, and finally conclusions are provided in \cref{sec:concl}.
\end{paragraph}

\section{Formal Preliminaries}
\label{sec:FP}

Throughout the paper, we assume a fixed finite  set $\agents\ne \varnothing$ of \emph{agents}.
As is common, we employ Kripke semantics to reason about agents' epistemic states. Since we are interested in the dynamics of belief chance, we use PAL, the logic of public announcements as the simplest version of DEL.

\begin{definition}[Language]
    The epistemic language with public announcements for agents from~$\agents$ is defined by
    \[
    \varphi \coloncolonequals p \mid \neg \varphi \mid (\varphi \land \varphi) \mid B_i \varphi \mid [\varphi] \phi
    \]
    where $p \in \mathit{Prop}$ is an \emph{atomic proposition} (or simply \emph{atom}) and $i \in \agents$.
\end{definition}

\begin{definition}[Kripke models] 
\label{def:kripke}
A \emph{Kripke model}  $\calM = \langle W, R, V\rangle$ is a triple comprising:
\begin{itemize}
    \item a  set of \emph{possible worlds}	 $W \ne \varnothing$;
    \item A function $R \colon \agents \to 2^{W\times W}$ that assigns to each agent $i \in \agents$ a binary relation $R(i) \subseteq W \times W$, called \emph{accessibility relation} and usually denoted~$R_i$ instead of~$R(i)$;
    \item a \emph{valuation function} $V \colon \atoms \to 2^W$ that assigns to each atom $p \in \atoms$ a set $V(p) \subseteq W$ of worlds  where $p$ holds.
\end{itemize}
We use notation $R_i(u) \colonequals \{v \in W \mid u R_i v\}$ for the set of all worlds that agent $i$ considers possible in a world~$u \in W$.
A \emph{pointed Kripke model} is a pair $(\calM, v)$ where $\calM$ is a Kripke model and  
 $v \in W$ represents the \emph{real} (or \emph{actual}) \emph{world}.
\end{definition}

\begin{definition}[Truth]
\emph{Truth of a formula $\phi$ at a world $w$ of a Kripke model $\calM= \langle W, R, V\rangle$} is defined recursively:  for atoms, $\calM, w \models p$ if{f} $w \in V(p)$; boolean connectives behave classically;
		$\calM, w \models B_i \varphi$  if{f}  $\calM, u \models \varphi$ for all  $u\in R_i(w)$;
finally, $\calM, w \models [\varphi] \psi$  if{f}  either $\calM, w \notmodels \phi$ or ($\calM, w \models \phi$ and $\calM\upd\varphi, w \models \psi$), where $\calM\upd\varphi \colonequals \langle W',R',V' \rangle$ is defined 
as:
\begin{itemize}
        \item $W' \colonequals \{u \in W \mid \calM, u \models \varphi\}$  (note that $w \in W'$ whenever this submodel is constructed);
        \item $R_i' \colonequals R_i \cap (W' \times W')$ for each $i \in \agents$;
        \item $V'(p) \colonequals V(p) \cap W'$ for each $p \in \atoms$.
    \end{itemize}  
    Strictly, speaking $\calM, w \models \phi$ and $\calM \upd \phi$ are defined by mutual recursion on $\phi$, as is standard in DEL.
		Formula $\varphi$ is \emph{false at world $w$}, denoted $\calM,w \notmodels \varphi$, if{f} it is not true at $w$.
\end{definition}

\begin{definition}[Binary relation types]
    A binary relation $R \subseteq W \times W$ is called  \emph{reflexive} if{f} $w R w$ for all \mbox{$w \in W$}; \emph{transitive} if{f} for all $w, v, u \in W$ we have $w R u$ whenever $w R v$ and $v R u$; \emph{euclidean} if{f} for all $w, v, u \in W$ we have $v R u$ whenever $w R v$ and $w R u$; \emph{symmetric} if{f} for all $w, v \in W$ we have $v R w$ whenever $w R v$. Relation $R$ is called an \emph{equivalence relation} if{f} it is reflexive, transitive, and euclidean; a \emph{partial equivalence relation} if{f} it is transitive and symmetric; an \emph{introspective relation} if{f} it is transitive and euclidean. \looseness=-1
\end{definition}
\begin{proposition}
    An equivalence relation is also  symmetric, hence, introspective and a partial equivalence relation. A partial equivalence relation is also euclidean, hence, introspective.
\end{proposition}
\begin{definition}[Model types]
    We call a Kripke model $\langle W, R, V\rangle$ \emph{epistemic} if{f} all $R_a$ are equivalence relations; \emph{introspective} if{f} all $R_a$ are introspective; \emph{quasi-epistemic} if{f} all~$R_a$~are partial equivalence relations. 
\end{definition}
\begin{proposition}
    An equivalence relation $R_a \subseteq W \times W$ partitions $W$ into equivalence classes, or \emph{$a$-clusters}, such that for each equivalence class  $E\subseteq W$ we have $u R_a v$ for any $u, v \in E$ and neither $u R_a u'$ nor $u' R_a u $ for any $u \in E$ and $u' \in W \setminus E$. A partial equivalence relation produces a similar partition but of a subset $W \setminus I$, where $I\subseteq W$ consists of isolated worlds, i.e., neither $i R_a w$ nor $w R_a i$ for any $w \in W$ and $i \in I$.
\end{proposition}
\begin{definition}[Composition and iteration of relations]
    For any binary relations $Q,Q' \subseteq W \times W$ on a set~$W$,  their \emph{composition} $Q \circ Q' \ce \{(w,v) \in W \times W \mid (\exists u \in W) ( w Q u \text{ and } u Q' v \}$.
Let $Q^k$ for  $k \geq 0$ be defined recursively by $Q^0 \ce \{(w,w) \mid w \in W\}$   and $Q^{k+1} \ce Q \circ Q^k$.
\end{definition}
\begin{definition}[Mutual and common accessibility]
For a Kripke model $\langle W, R, V\rangle$, we define the  
\begin{itemize}
    \item  \emph{mutual accessibility relation} $R_\agents \ce \bigcup_{a \in \agents} R_a$  that corresponds to the mutual belief of all agents;
    \item \emph{common accessibility relation} $R^*_\agents \ce \bigcup_{k=0}^{\infty} R^k_{\agents}$ that corresponds to the common belief of all agents.
\end{itemize}

\end{definition}

\begin{definition}[Agent's submodel]
\label{def:submodel}
    Let $(\calM,v)$ with $\calM = \langle W, R, V\rangle$ be a pointed model and $i\in \agents$ be an agent with \emph{consistent beliefs}, i.e., such that $R_i(v) \ne \varnothing$. The \emph{submodel accessible by $i$ at\/ $(\calM,v)$}, or \emph{$i$'s~part/submodel of\/ $(\calM,v)$}   is the Kripke model $\calM_v^i \colonequals \langle W', R', V'\rangle$ such that
    \begin{itemize}
        \item $W'\ce \left(R_i \circ R^*_\agents\right)(v)$ is the set of all worlds  that are common-belief accessible from any point in~$R_i(v)$, including all worlds from $R_i(v)$;
        \item $R_j' \colonequals R_j \cap (W' \times W')$ for each $j \in \agents$;
        \item $V'(p) \colonequals V(p)\cap W'$ for each $p \in \atoms$.
    \end{itemize}    
   For reasons of uniformity, we sometimes abuse the terminology and say that  $i$'s submodel is empty when $R_i(v)= \varnothing$.
\end{definition}

One of the key features of distributed systems is the distinction between the global state of the system (represented by the actual world) and the local state of agents, that is usually a limited portion of the global state. To represent this distinction, we typically consider pointed models $(\calM,v)$, where the real world $v$ is not considered possible by any of the agents.
This represents the \textit{fallibilistic} property of fault-tolerant distributed systems that  no agent can be sure to be fault-free and, hence, cannot be sure to perceive the global state of the system correctly~\cite{KuznetsPSF19FroCoS}. This modeling choice does not preclude agents from 
having an accurate view of the system during execution since the possibilities they consider can include a faithful duplicate of the actual world.
More strikingly, nothing prevents the global state of the system to be completely different from what an agent has in the local view, e.g.,~in case of corrupted memory. 
By the same token, not all worlds are equally possible from the local viewpoint of an agent: while some worlds are actual candidates for being the real world, for a given agent, other worlds are only \textit{virtually possible}, i.e.,~they are used to compute higher-order beliefs about other agents. This distinction is crucial in epistemic puzzles: e.g.,~in the more traditional  consecutive numbers problem (\cref{ex:consec}), $a$ considers only the pairs where his number is~$1$ to be actually possible: all other pairs are present in order for him to correctly compute his higher order beliefs about~$b$.
In this paper, we propose an a~priori belief update mechanism that incorporates a fallibilist assumption and also sharply separates, for the updated agent(s), actually possible worlds from virtually possible worlds.\looseness=-1

\section{A~priori Beliefs in Epistemic Logic}
\label{sec:APB_EL}

It is important to differentiate between a~priori beliefs that can and cannot be represented syntactically. We call the former \emph{explicit a~priori beliefs} and define them as beliefs that can be represented by one epistemic formula in the object language. Thus, beliefs that can be represented by a finite set of formulas are explicit because the conjunction of the set provides an equivalent description. 
However, not all a~priori beliefs are explicit. For instance, factivity of  everyone's knowledge is usually described as either a frame property (reflexivity) or an infinite set of formulas ($B_a \phi \to \phi$ for all $a$ and $\phi$) but cannot be reduced to the truth of one formula only. We call such a~priori beliefs  \emph{implicit}. 

In as far as epistemic puzzles and distributed systems deal with a~priori beliefs (usually without calling them that), it is done by providing the initial Kripke models that agents proceed to modify based on the information they receive. In other words, the implicit form of a~priori beliefs is dominant, and converting it into an explicit representation may well be an unrealistic task even for simple epistemic puzzles~\cite{Artemov22}.

\subsection{Privacy of A~priori Beliefs}
\label{subsec:priv}
In the typical modeling of epistemic puzzles,  agents are assumed to have common factual a~priori beliefs, as represented by the commonly known epistemic model~\cite{artemov}. This ensures the homogeneity of reasoning by all agents and creates no problems as long as agents have no need to modify their a~priori beliefs. At the same time, for our setting, it is unreasonable to assume that the internal reasoning process leading one agent to modify its  a~priori beliefs would be noticeable by other agents, let alone be commonly known among them, even in cases where they start with commonly known priors. 

Hence, we abandon the assumption of the common knowledge of the model and treat each agent as having its own private a~priori beliefs represented by the submodel of this agent (cf.~\cref{def:submodel}). If submodels of several agents overlap, this is treated as a coincidence, especially in view of the fact that each agent, being only aware of its own submodel, would be oblivious to any such overlaps.

Note that if the initial model is a connected epistemic model, the submodel of each agent is the whole model. We interpret this situation as each agent believing to have the common model that they have, but leave the possibility of one or several agents being wrong, typically as a result of a later a~priori public update.\looseness=-1

To summarize, (i)~agents' reasoning is represented by a pointed Kripke model $(\calM,v)$, but instead of assuming common knowledge of $(\calM,v)$, an agent $i$ is only guaranteed to know $i$'s submodel $\calM_v^i$, which may or may not be different from submodels $\calM_v^j$ visible by other agents $j$; and (ii) in addition to the standard public update mechanism for public announcements, we employ a \emph{private} model update mechanism for a~priori belief updates of an agent $i$ that results in $i$'s submodel becoming disjoint from other agents' submodels after $i$ performed an a~priori update.

\subsection{A~priori Belief Updates as Conflict Resolution}
\label{subsec:confl}
    We still maintain agents' reasoning within a pointed Kripke model, i.e., each agent $i$ is still logically omniscient w.r.t.~the a~posteriori information contained in $i$'s submodel~$\calM_v^i$ of the given pointed Kripke model~$(\calM,v)$. 
    The observations that the agent $i$ is making during the run may come into  conflict with its  a~priori beliefs, which would manifest as $R_i(v)=\varnothing$ causing $i$'s beliefs to become inconsistent.
    As long as such contradiction does not arise, $i$ continues a~posteriori reasoning in the standard epistemic manner or DEL manner for publicly announced information. In light of the Knowledge of Preconditions Principle~\cite{Mos15TARK}, which states that $B_i \phi$ is a precondition for an action whenever $\phi$ is, inconsistent beliefs make it impossible for the agent to act correctly. Indeed, even if the agent is supposed to choose between mutually exclusive actions $A$ if $\phi$ holds or  $B$ otherwise, the inconsistent agent would have to perform both due to believing everything, including both $\phi$ an $\neg \phi$. This provides a good incentive for the agent to reexamine its a~priori beliefs and 
    try updating  them in such a way as to restore the consistency. 
    Semantically, this is achieved by~$i$ creating a new submodel $\calM_v^i$ (disjoint from the existing model) for itself in  place of the current empty one.
    This is exactly the desired functionality for SASO systems: if an agent finds out during runtime that its a~priori beliefs are inadequate, i.e., if an a~posteriori epistemic update of an agent violates some of its a~priori beliefs, the agent may initiate a~priori reasoning, aiming at finding new a~priori beliefs that comply with the 
    a~posteriori epistemic status. If an appropriate solution is found, an \emph{a~priori belief update} is privately invoked for installing a new private submodel within the current Kripke model.

Needless to say, there are many conceivable ways for resolving conflicts arising from inconsistencies in SASO systems, mainly because, in accordance with the Duhem--Quine thesis~\cite{quine}, it is not always possible to isolate the specific hypothesis (a~priori belief in our case) as the culprit for that inconsistency,
even for a restricted set of explicit a~priori beliefs. 
Thus, we will not try to provide deterministic algorithms for choosing an appropriate a~priori update. At the same time, we do not leave this process completely \emph{ad~hoc}. A new submodel is constructed from several building blocks representing the agent's new guesses regarding  the actually possible worlds, the virtually possible worlds, and their relationship, as described in~\cref{sec:APBupdate}.\footnote{Providing a Kripke model synthesis procedure to produce actually and virtually possible worlds to satisfy specific explicit a~priori assumptions is left for future~work.}\looseness=-1

\subsection{Higher-Order A~priori Reasoning}
\label{subsec:hoAPR}
The question of higher-order belief updates remains outside the scope of this paper. In other words, if an agent $a$ detects some inconsistency in the beliefs of another agent $b$, we do not force $a$ to try and guesstimate an a~priori belief update by $b$, even if agent $b$ is, in $a$'s estimation, likely to perform it.
Since a~priori updates are likely to be be non-deterministic, there is little reason to assume that $a$ would be able to exactly match the thought processes of another agent $b$. 
In Lewis's terminology, once the belief in shared ``inductive standards and background information'' fails, so does the ground for a common understanding of the situation.
In practice, this means that $a$ would generally lose the ability to interpret $b$'s actions or gain information from them\footnote{It might be possible to actually communicate an agent's a~priori beliefs, but of course, such a communication action would correspond to an a~posteriori update  concerning a~priori beliefs. Modeling this complex interaction is outside the scope of this paper.}. In effect, from this moment on, $a$ would treat $b$ as a fully byzantine agent.

\section{A~priori Belief Updates}
\label{sec:APBupdate}
The aim of this section is to describe the semantic mechanism an agent $a$ can use to perform a~priori belief updates, in order to attempt self-recovery when $a$ discovers that its a~posteriori observations came into conflict with its a~priori assumptions. Note that this means that $a$'s part of the model is empty, making it impossible for~$a$ to use  the current pointed model to recover a consistent epistemic state. Much like epistemic puzzles are described by semantic models, the a~priori belief updates are also semantic: agent $a$ tries  to reimagine the epistemic situation as a \emph{trial model} that is based on the agents' previous experiences, modifications of explicit a~priori beliefs, and/or \emph{ad~hoc} guesses. However, $a$ generally has no reason to ascribe these internal attempts to the thinking of other agents. Hence, to model higher-order beliefs for all other agents,  $a$ uses  some \emph{backup model} that is typically derived from $a$'s previous experiences and knowledge of a~priori beliefs of other agents.\footnote{An important exception to this rule is the scenario where $a$ suspects to be the only one to have erred. Then it would be natural for~$a$ to try to adapt itself to the alleged thinking of other agents, using that for both  trial and backup models.}

\subsection{General Update Mechanism}
\label{subsec:generalmech}

We first formulate this update mechanism to be as general as possible, requiring only the basic coherency restrictions on the trial and backup models. In \cref{subsec:heuristics}, we will discuss some strategies autonomous agents may employ to generate these models. We do stress, however, that trial models are intended to be \emph{ad~hoc} guesses. The amount of restrictions on  the trial model  should be inversely proportional to how wrong agent's beliefs are expected to be: the further away from reality the agent may have strayed, the fewer restrictions should be imposed on its imagination in the process of recovering a consistent state.

For instance, in the consecutive numbers puzzle (\cref{ex:consec}), the trial model should not include non-integer numbers or violate laws of arithmetic. On the other hand, if agents are expected to not be fully attentive to the formulation of the puzzle, a trial model may include pairs of numbers that are not consecutive to account for possible misunderstandings.

\begin{definition}[A~priori belief update]
\label{def:APBU}
An \emph{a~priori belief update} for agent $a$ is a tuple 
\begin{equation}    
\label{eq:update}
\U=(\calM^a, U^a, \calM^{\neg a}, \C)
\end{equation}
where
\begin{itemize}
    \item \emph{trial model} $\calM^a = \langle W^a, R^a,  V^a\rangle$ is a quasi-epistemic  Kripke model,
\item  $U^a \subseteq W^a$ is an $a$-cluster within $\calM^a$  (note that $U^a \ne \varnothing$),
\item \emph{backup  model} $\calM^{\neg a} = \langle W^{\neg a}, R^{\neg a}, V^{\neg a}\rangle$ is a quasi-epistemic Kripke model,
\item $S \C W^{\neg a}$ is a \emph{correspondence function} from some subset $S \subseteq W^a$ of the domain of the trial model into the domain of the backup model, i.e., a partial function from~$W^a$~to~$W^{\neg a}$ that identifies some of the trial worlds with backup worlds
\end{itemize}
such that for any atoms $p \in \atoms$ and for any trial worlds $u,v \in W^a$ and backup worlds $u',v' \in W^{\neg a}$  the following \emph{coherency conditions} are fulfilled:
\begin{enumerate}
    \item \emph{atomic coherency}: if $u \C u'$, then $u \in V^a(p)\quad \Longleftrightarrow \quad u' \in V^{\neg a}(p)$, i.e., only propositionally equivalent worlds can be identified;
    \item \emph{reasoning coherency}: for each agent $b \ne a$, if $u \C u'$ and $v \C v'$, then $u R^a_b v \quad \Longleftrightarrow \quad u' R^{\neg a}_b v'$, i.e., the identification respects indistinguishabilities of all agents but $a$;
    \item \emph{simulation coherency}: for each agent $b \ne a$, if $u \C u'$ and $u'R^{\neg a}_b v'$, then there exists $v \in W^a$ such that $u R^a_b v$ and $v \C v'$, i.e., the trial model simulates the backup model for all  agents but~$a$.
    \end{enumerate}
\end{definition}

Intuitively, $U^a$ represents the local state of $a$, i.e., those worlds that $a$ considers actually possible.
Relations $R^a_b$ for $b \ne a$ determine whether these new worlds constructed by $a$ would have been distinguishable for other agents, were they aware of  $a$'s new trial vision of the world. However, since they are unaware of $a$'s trial model $\calM^a$, these indistinguishabilities need to be transferred to the backup model $\calM^{\neg a}$, which represents $a$'s
virtually possible worlds, used by $a$ to understand  how the other agents imagine the epistemic situation.
The coherency conditions on $\C$ ensure 
compatibility between $a$'s actually possible worlds and the virtually possible worlds considered by $a$.
In particular, the trial model should simulate the backup model for agents other than $a$ because  $a$'s  understanding of their epistemic state cannot be worse than $a$'s~impression of their own understanding.

The result of applying an a~priori belief update to a pointed Kripke model is described by the following definition:

\begin{definition}[Result of a~priori belief update]
\label{def:result:APBU}
    Let $\U=(\calM^a, U^a, \calM^{\neg a}, \C)$ be an a~priori belief update with $\calM^a = \langle W^a, R^a, V^a\rangle$ and $\calM^{\neg a} = \langle W^{\neg a}, R^{\neg a}, V^{\neg a}\rangle$,  and let $(\calM,v)$ be a pointed Kripke model with introspective model $\calM = \langle W, R, V\rangle$ such that $R_a(v) = \varnothing$.
    The result of agent $a$ applying  
    $\U$ to  $(\calM,v)$ is a pointed Kripke model 
    $(\calM\circledcirc_a\U ,v)$ where  $\calM\circledcirc_a\U\colonequals \langle W', R', V' \rangle$ such that:
    \begin{itemize}
    \item $ W' \colonequals W \sqcup U^a \sqcup W^{\neg a}$;
    \item $R'_a \colonequals R_a  \sqcup R_a^{\neg a}\sqcup  (\{v\}\times U^a) \sqcup (U^a \times U^a) $;
    \item  
  $      R'_b \colonequals R_b \sqcup R_b^{\neg a} \sqcup \{(u, v') \in U^a \times W^{\neg a} \mid (\exists v\in W^{a}) u R^a_b v \C v'\}$ for each agent $b\ne a$;
    \item $V'(p) \colonequals V(p) \sqcup V^{\neg a}(p) \sqcup \bigl(V^a(p)\cap U^a\bigr) $\qquad for any $p \in \mathit{Prop}$.
    \end{itemize}
\end{definition}

\begin{remark}
    The requirement of simulation coherency is asymmetric. It is worth asking why  it should not be a bidirectional bisimulation instead. Recall that both  trial and backup models describe $a$'s~attempt to provide an alternative explanation for the apparent inconsistency of its beliefs. The difference is that the trial model~$\calM^a$ is $a$'s~suggestion of how things ``actually~are'' whereas the backup model~$\calM^{\neg a}$ is $a$'s~attempt to imagine how  this new picture of the world is viewed by all other agents based on some standard, default, or past views $a$ expects of them. This puts $a$ in a position of intellectual superiority akin to that of the system designer of a distributed system. Within this new point of view of $a$'s creation, agent $a$ \emph{does} know better than the other agents, much like the system designer knows better than distributed agents.\footnote{Note that this refers to the knowledge of the underlying model rather than belief in specific facts. In fact, if $a$ considers more possibilities, $a$ would believe fewer facts, presumably because contracting beliefs was necessary to avoid inconsistency.} Indeed, as mentioned before, it is the  system designer who performs a~priori belief updates in traditional distributed systems, which makes it reasonable for an agent to adapt a similar attitude for the same task.  In $a$'s view, $\calM^a$ is an ``improved'' description of the world, which $a$ hopes to be accurate, whereas $\calM^{\neg a}$ may well be an erroneous description that $a$~is aware of but rejects. If $a$~thinks that other agents should consider certain possibilities, it is only rational for~$a$ to embed these possibilities into its own  view of the world to ensure its accuracy by better describing the epistemic reasoning of other agents. At the same time, $a$~can imagine others having blind spots and missing some possibilities $a$~is now considering. A simulation in the opposite direction, from~$\calM^{\neg a}$ to~$\calM^a$ would have made such blind spots impossible and, hence, is not required by our definition. At the same time, if $a$'s update is guided by the idea that other agents knew the true state of affairs all along and it is $a$ who has been missing something, then the simulation in the opposite direction would make sense. Hence, while not requiring it, our definition does not precludes a bisimulation between~$\calM^a$~and~$\calM^{\neg a}$ either.
\end{remark}

The properties of a~priori belief updates are captured by the following theorem:

\begin{theorem}
\label{Th:cons_Bel}
    Let\/ $\U=(\calM^a, U^a, \calM^{\neg a}, \C)$ be an a~priori belief update 
and\/  $(\calM,v)$ be a pointed Kripke model with introspective model $\calM = \langle W, R, V\rangle$ such that $R_a(v) = \varnothing$. Then:
\begin{enumerate}
    \item $\calM\circledcirc_a\U$ is an introspective model.
    \item For any purely propositional formula $\phi$, 
    \[
    \calM, v \models \phi \qquad\Longleftrightarrow\qquad \calM\circledcirc_a\U, v \models \phi.
    \]
    \item\label{ind_beliefs} For any  formula $\phi$ and any agent $b \ne a$,
    \[
    \calM, v \models B_b\phi \qquad\Longleftrightarrow\qquad \calM\circledcirc_a\U, v \models B_b\phi.
    \]
    \item\label{consist} $\calM, v \models B_a\bot$,\qquad but\qquad $\calM\circledcirc_a\U, v \notmodels B_a\bot$.
\end{enumerate}
\end{theorem}

\begin{proof}
Let 
$\calM\circledcirc_a\U= \langle W', R', V' \rangle$.
\begin{enumerate}
    \item Transitivity and euclideanity of $R'_a$ and $R'_b$ for $b \ne a$ follow by construction due to all involved models being introspective. We demonstrate only several  non-trivial cases. (a) It cannot happen that $w R'_a v$ and $v R'_a u$ for some $u \in U^a$ because,  by construction,  $w R'_a v$ is equivalent to $w R_a v$, which  would imply $v R_a v$ by euclideanity of $R_a$, whereas we assumed $R_a(v)= \varnothing$. (b) Let $b \ne a$. If $u R'_b v'$ because $u R^a_b v \C v'$ and $v' R^{\neg a}_b z'$ for some $u \in U^a$, $v \in W^a$, and $v',z' \in W^{\neg a}$, then there exists $z \in W^a$ such that $v R^a_b z \C z'$ by the simulation coherency condition. Hence, $u R^a_b z$  by transitivity of $R^a_b$, which implies $u R'_b z'$, as required. (c)~Let~$b \ne a$. If $u R'_b v'$ because $u R^a_b v \C v'$ and $u R'_b z'$ because $u R^a_b z \C z'$ for some $u \in U$, $v,z \in W^a$, and $v',z' \in W^{\neg a}$, then $v R^a_b z$ by euclideanity of $R^a_b$. Hence, $v' R^{\neg a}_b z'$ by the reasoning coherency condition, as required.
    \item Easily follows from the fact that the propositional valuation at $v$ does not change.
    \item Easily follows from $R'_b(v) = R_b(v)$.
    \item Easily follows from the fact that $R'_a(v) = U^a \ne \varnothing$.\qedhere
\end{enumerate}
\end{proof}

\begin{remark}
    Note that Theorem~\ref{Th:cons_Bel}.\ref{ind_beliefs} means that $a$ performing a successful update of its a~priori beliefs cannot resolve inconsistent beliefs of other agents (though $a$ may erroneously believe to have resolved them). It is reasonable to expect each of the agents with inconsistent beliefs to perform such an operation. We will describe how to do it and why updates of different agents are completely independent from each other in \cref{subsec:example_Simultaneous}.
\end{remark}

\subsection{Application to the Muddy Children Puzzle}
\label{subsec:examples}
We now illustrate this a~priori update mechanism by considering   the muddy children puzzle~\cite{bookof4} with various additional  a~priori assumptions.

\begin{example}[Standard muddy children puzzle, MCP]
    \label{MCEX:1}
In the muddy children puzzle, $n$~children are playing in the mud, and $k$~of them get mud on their foreheads. Each can see whether there is  mud  on others but not on his/her own forehead. Father comes and announces, ``At least one of you has muddy  forehead.'' Father then start repeating the question, ``If you know whether you are muddy or not, step forward.'' Under the assumption that children are perfect reasoners, pay complete attention, are truthful, and all this (and much more) is common knowledge among them, it is well known that the first $k-1$ times Father asks, nobody steps forward. All~$k$~muddy children step forward the $k$th time Father asks, and the remaining $n-k$ clean children step forward the next $k+1$th time.
For instance, for $n=3$ and $k=1$, i.e., with three children playing and one muddy child, the muddy child should immediately step forward, with the other two children stepping forward the second time.

\begin{figure}
\begin{subfigure}{.99\textwidth}
 \centering
  \resizebox{4cm}{!}{%
 \begin{tikzpicture}
\node[nod] at (0,0) [label=below:$AB$] (AB) {};
\node[nod] at (5,0) [label=below:$A$] (A) {};
\node[nod] at (0,4) [label=above:$ABC$] (ABC) {};
\node[nod] at (5,4) [label=above:$AC$] (AC) {};
\node[nod] at (2,2) [label=left:$B$] (B) {};
\node[nod] at (7,2) [label=below:$0$] (000) {};
\node[nod] at (2,6) [label=above:$BC$] (BC) {};
\node[nod] at (7,6) [label=above:$C$] (C) {};
\node [left=1.5cm of BC] {\textbf{\huge{$\calM_0$}}};
\path[<->,{Stealth[scale width=1.5]}-{Stealth[scale width=1.5]}] (ABC)
edge node[above] {$b$} (AC)
edge node[above] {$a$} (BC)
edge node[right] {$c$} (AB)
(A)
edge node[above] {$b$} (AB)
edge node[above] {$a$} (000)
edge node[pos=0.4, right] {$c$} (AC)
(C)
edge node[right] {$c$} (000)
edge node[above] {$a$} (AC)
edge node[above] {$b$} (BC);
\path[dashed, <->,{Stealth[scale width=1.5]}-{Stealth[scale width=1.5]}] (B)
edge node[above] {$b$}(000)
edge node[above] {$a$} (AB)
edge node[pos=0.4, right] {$c$} (BC);
\end{tikzpicture}
}
  \caption{Epistemic model $\calM_0$ for the muddy children puzzle with three agents of \cref{MCEX:1}, before Father's announcements. $D$ in the name of the state corresponds to~$m_d$ being true at the state, i.e., agent $d$ being muddy. No child is muddy in state~$0$. Bidirectional arrows (both dashed and solid) represent indistinguishability for agents, e.g., agent $b$ cannot distinguish between states $ABC$ and $AC$. Reflexive loops, present for all agents at every state, are omitted.}
    \label{fig:EX1-init&MB}
\end{subfigure}

\par\bigskip 

\begin{subfigure}{.99\textwidth}
\resizebox{13cm}{!}{
   \begin{tikzpicture}
\node[nod] at (-10,0) [label=below:$AB$] (110') {};
\node[nod] at (-10,4) [label=above:$ABC$] (111') {};
\node[nod] at (-5,4) [label=above:$AC$] (101') {};
\node[nod] at (-8,6) [label=above:$BC$] (011') {};

\node[nod, fill=gray] at (-5,0) [label=below:$\underline{A}$] (@) {};
\node[nod] at (-1,0) [label=below:$AB$] (AB) {};
\node[nod] at (4,0) [label=below:$A$, label=135:\large{$U^a$}] (A) {};
\node[nod] at (-1,4) [label=above:$ABC$] (ABC) {};
\node[nod] at (4,4) [label=above:$AC$] (AC) {};
\node[nod] at (1,2) [label=left:$B$] (B) {};
\node[nod] at (1,6) [label=above:$BC$] (BC) {};
\node[nod] at (6,6) [label=above:$C$] (C) {};
\node[nod] at (9,0) [label=below:$AB$] (110) {};
\node[nod] at (9,4) [label=above:$ABC$] (111) {};
\node[nod] at (14,4) [label=above:$AC$] (101) {};
\node[nod] at (11,6) [label=above:$BC$] (011) {};
\node (p) [left=2.6cm of 011'] {};
\node (p1) [above=0.8cm of p] {\textbf{\Large{$(\calM, \underline{A})$}}};
\node [right=7.5cm of p1] {\textbf{\Large{$\calM^a$}}};
\node [right=17.5cm of p1] {\textbf{\Large{$\calM^{\neg a}$}}};
\node (r) [above=0.15cm of p] {};
\node (r1) [right=0.2cm of r] {\textbf{\Large{$W$}}};
\node [right=8cm of r1] {\textbf{\Large{$W^a$}}};
\node [right=22.5cm of r1] {\textbf{\Large{$W^{\neg a}$}}};

\node (L) [right=7cm of p1] {};
\node (L1) [below=8cm of L] {};

\path[<->, {Stealth[scale width=1.5]}-{Stealth[scale width=1.5]}]
(ABC)
edge node[above] {$a$} (BC)
edge node[right] {$c$} (AB)
(C)
edge node[above] {$a$} (AC)
edge node[above] {$b$} (BC)
(B)
edge node[above] {$a$} (AB)
edge node[right, , pos=.3] {$c$} (BC)
(AC)
edge node[above] {$b$} (ABC)
edge node[right] {$c$} (A)
(A)
edge node[above] {$b$} (AB)
(111)
edge node[above] {$b$} (101)
edge node[above] {$a$} (011)
edge node[right] {$c$} (110)
(111')
edge node[above] {$b$} (101')
edge node[above] {$a$} (011')
edge node[right] {$c$} (110');

\path[->, -{Stealth[scale width=1.5]}] (@)
edge[] node[right] {$c$} (101')
edge[] node[above] {$b$} (110');

\path[->,-{Stealth[scale width=1.5]}] 
          (011) edge[loop left] node[left] {$a,b,c$} (011)
          (111) edge[loop left] node[left] {$a,b,c$} (111)
          (110) edge[loop left] node[left] {$a,b,c$} (110)
          (101) edge[loop below] node[below] {$a,b,c$} (101);
\path[->,-{Stealth[scale width=1.5]}] 
          (BC) edge[loop left] node[left] {$a,b,c$} (BC)
          (ABC) edge[loop left] node[left] {$a,b,c$} (ABC)
          (AB) edge[loop left] node[left] {$a,b,c$} (AB)
          (A) edge[loop right] node[right] {$a,b,c$} (A)
          (B) edge[loop right] node[right] {$a,b,c$} (B)
          (C) edge[loop right] node[right] {$a,b,c$} (C)
          (AC) edge[loop right] node[right] {$a,b,c$} (AC);
\path[->,-{Stealth[scale width=1.5]}] 
          (011') edge[loop left] node[left] {$a,b,c$} (011')
          (111') edge[loop left] node[left] {$a,b,c$} (111')
          (110') edge[loop left] node[left] {$a,b,c$} (110')
          (101') edge[loop right] node[right] {$a,b,c$} (101');

\path[dashed]
(ABC) edge[bend left=15] node[] {} (111)
(BC) edge[bend left=15] node[] {} (011)
(AC) edge[bend right=15] node[] {} (101)
(AB) edge[bend right=15] node[] {} (110)
;

\path[line width=0.5]
(L) edge node[] {} (L1);

\pgfdeclarelayer{background layer}
\pgfdeclarelayer{background layer}
\pgfsetlayers{background layer, background layer, main}

\begin{pgfonlayer}{background layer}
      
\filldraw [line width=15mm,line join=round,black!10]
      (BC.north  -| C.east)  rectangle (AB.south  -| AB.west);
      
\filldraw [line width=15mm,line join=round,black!10]
      (011.north  -| 101.east)  rectangle (110.south  -| 110.west);

\filldraw [line width=15mm,line join=round,black!10]
      (011'.north  -| @.east)  rectangle (110'.south  -| 110'.west);
\end{pgfonlayer}

\begin{pgfonlayer}{background layer}
    \filldraw [line width=12mm,line join=round,black!20]
      (A.north  -| A.east)  rectangle (A.south  -| A.west);
\end{pgfonlayer}

\end{tikzpicture}
}
   \caption{Left: Initial pointed Kripke model $\calM,\underline{A}$ restricted to $\APB_2$.\\ 
   Middle and Right: Elements of a~priori belief update $\U = (\calM^a, U^a, \calM^{\neg a}, \C)$ for agent $a$ from \cref{MCEX:2}. The correspondence relation $\C$ is represented by dashed lines.}
    \label{fig:EX1.1}
\end{subfigure}

\par\bigskip 

\begin{subfigure}{.99\textwidth}
\resizebox{13cm}{!}{%
     \begin{tikzpicture}
\node[nod] at (-10,0) [label=below:$AB$] (110') {};
\node[nod] at (-10,4) [label=above:$ABC$] (111') {};
\node[nod] at (-5,4) [label=above:$AC$] (101') {};
\node[nod] at (-8,6) [label=above:$BC$] (011') {};

\node[nod, fill=gray] at (-5,0) [label=below:$\underline{A}$] (@) {};

\node[nod] at (4.5,0) [label=below:$A$, label=135:\Large{$U^a$}] (A) {};

\node[nod] at (9,0) [label=below:$AB$] (110) {};
\node[nod] at (9,4) [label=above:$ABC$] (111) {};
\node[nod] at (14,4) [label=above:$AC$] (101) {};
\node[nod] at (11,6) [label=above:$BC$] (011) {};

\node (p) [left=2cm of 011'] {};
\node (p1) [above=0.2cm of p] {\textbf{\Large{$W$}}};
\node [right=18cm of p1] {\textbf{\Large{$W^{\neg a}$}}};
\node [above=0.5 of p1] {\textbf{\Large{$(\calM \circledcirc_a \U, \underline{A})$}}};

\path[<->, {Stealth[scale width=1.5]}-{Stealth[scale width=1.5]}]
(111)
edge node[above] {$b$} (101)
edge node[above] {$a$} (011)
edge node[right, , pos=.4] {$c$} (110)
(111')
edge node[above] {$b$} (101')
edge node[above] {$a$} (011')
edge node[right] {$c$} (110');

\path[->, -{Stealth[scale width=1.5]}] 
(@)
edge[] node[right] {$c$} (101')
edge[] node[above] {$b$} (110')
edge[] node[above] {$a$} (A)
(A)
edge[] node[below, pos=.6] {$c$} (101)
edge[] node[above] {$b$} (110)
;

\path[->,-{Stealth[scale width=1.5]}] 
          (011) edge[loop left] node[left] {$a,b,c$} (011)
          (111) edge[loop left] node[left] {$a,b,c$} (111)
          (110) edge[loop right] node[right] {$a,b,c$} (110)
          (101) edge[loop below] node[below] {$a,b,c$} (101)
          (011') edge[loop left] node[left] {$a,b,c$} (011')
          (111') edge[loop left] node[left] {$a,b,c$} (111')
          (110') edge[loop left] node[left] {$a,b,c$} (110')
          (101') edge[loop right] node[right] {$a,b,c$} (101')
          (A) edge[loop above] node[above] {$a$} (A)
          ;

\begin{pgfonlayer}{background}

\filldraw [line width=12mm,line join=round,black!10]
      (A.north  -| A.east)  rectangle (A.south  -| A.west);
      
\filldraw [line width=15mm,line join=round,black!10]
      (011.north  -| 101.east)  rectangle (110.south  -| 110.west);

\filldraw [line width=15mm,line join=round,black!10]
      (011'.north  -| @.east)  rectangle (110'.south  -| 110'.west);
\end{pgfonlayer}

\end{tikzpicture}
}
   \caption{Updated Kripke model $(\calM\circledcirc_a\U ,\underline{A})$ from \cref{MCEX:2}. In particular, arrows from $U^a$ to worlds in $W^{\neg a}$ are drawn in accordance to the correspondence relation $\C$, and all worlds in $U^a$ are reachable from $\underline{A}$ via $a$-arrows as per \cref{def:result:APBU}.}
    \label{fig:EX1.3}
\end{subfigure}
\caption{}
\end{figure}

We will use the standard epistemic modeling where the children are  agents $a, b, c, \ldots $. The muddiness of child $a$ is represented by atom $m_a$ that is true if{f} $a$ is muddy.  The fact that, before Father's first announcement, it is common knowledge that children do not know their own state but know that of others is formalized by requiring the validity of $\lnot B_a m_a \land \lnot B_a \neg m_a$ for all children $a$ and of $B_a m_b \lor B_a \neg m_b$ for  $a \ne b$.
A Kripke model satisfying these and other requirements of the puzzle
for the three agents is represented in~\cref{fig:EX1-init&MB}. Father's role is that of an external announcer, making his epistemic state irrelevant to the solution of the puzzle. From a distributed perspective, his role is akin to that of the system designer. While typically modeled as a public announcement, his first announcement ``at least one of you has muddy forehead''  can also be viewed as 
an a~priori belief update of the system performed by the system designer.
\end{example}

\begin{example}[MCP with false a~priori assumption resolved by  a~priori belief update]
\label{MCEX:2}
Consider a variant of the muddy children puzzle where (after playing the  same game every day) three children $a$, $b$, and $c$ have come to the conclusion that at least two of them always get muddy, i.e.,~they formed a common a~priori belief \looseness=-1
\[
\APB_2=(m_a \land m_b) \lor (m_a \land m_c) \lor (m_b \land m_c).
\]
(Note that $\APB_2$ does not represent \emph{all} common a~priori beliefs of the children. Most of those are encoded in  epistemic model $\calM_0$.)

In terms of Kripke modeling, this can be viewed as a public announcement of $\APB_2$, which results in  model $\calM_0$ from \cref{fig:EX1-init&MB} they used to commonly consider as the base model shrinking to  the Kripke model of four worlds only ($AB$, $AC$, $BC$, and $ABC$), making it unnecessary for Father to announce anything.   Suppose that this time around, contrary to habit, only $a$ is muddy, i.e., $m_a \land \neg m_b \land \neg m_c$. Thus, their common a~priori belief $\APB_2$ turns out to be false.  This scenario is depicted as pointed model $(\calM, \underline{A})$ in \cref{fig:EX1.1} (Left) that is obtained from the four-world model by adding a real world $\underline{A}$\footnote{It is named $A$ because its propositional valuation is the same as that of $A$ in $\calM_0$, while the underline means it is the actual world.} and drawing arrows from it to the four worlds in accordance with what children would consider possible. Note that there are now one-directional arrows due to the fact that children have false beliefs as the result of false a~priori assumptions.

While beliefs of all agents are false, only $a$ can detect this because $a$'s beliefs are inconsistent, unlike those of $b$ and $c$: 
 \[\calM,\underline{A} \models B_a \bot, \text{ while } \calM,\underline{A} \notmodels B_b \bot \text{ and } \calM,\underline{A} \notmodels B_c \bot.
 \]
 Thus, $a$ is the only child who is justified to update her a~priori beliefs. Note that no public announcement was made, and that $a$'s inconsistency is the result of $a$'s reasoning about the epistemic situation.

The introduction of one-directional arrows changes the whole interpretation of the epistemic situation. 
Indeed, using the assumption of the common knowledge of the model, it could be tempting to suggest that agents $b$ and $c$ can detect $a$'s inconsistency. 
It would not, however, match the underlying scenario. 
For instance, when $b$ sees that $a$ is muddy while $c$ is not, we would expect $b$ to conclude that $b$ is the second muddy child and that $a$ sees the mud on $b$'s forehead (the model says as much by providing a unidirectional $b$-arrow from $\underline{A}$ to $AB$). 
To match this intuition, we are forced to abandon the assumption of common knowledge of the model. 
For pointed model $(\calM,\underline{A})$, agent $b$ is only aware of the world $AB$, the only possible for $b$, and all worlds accessible from it by a sequence of arrows, agent $c$ is only aware of the world~$AC$, the only possible for $c$, and all worlds accessible from it by a sequence of arrows, while $a$ is not aware of any worlds, resulting in inconsistent beliefs. In particular, $a$ has lost the ability to examine the reasoning of other agents.
While it is natural for $a$ to assume that they still operate within the last commonly considered model, i.e., the four-world $\calM^{\neg a}$,  connecting this assumption to $a$'s own reality requires this reality to be fleshed out by~$a$, which is the purpose of $\calM^a$ (and, of its restriction $U^a$).

Let $\calM^a$ be as depicted in \cref{fig:EX1.1} (Middle), which fits well  with what she observes where the singleton cluster  $U^a=\{A\}$ of worlds from $\calM^a$  represents the possibilities $a$ considers actually possible based on her observations. As already mentioned,  for the backup model, which $a$ uses for computing what is considered by $b$ and $c$, agent $a$ takes the four-world model $\calM^{\neg a}$ she herself used until recently.
The elements of the a~priori belief update $\U = (\calM^a, U^a, \calM^{\neg a}, \C)$ are depicted in \cref{fig:EX1.1} (Right).
In particular, since $\calM^a$ extends $(\calM,\underline{A})$ with three additional worlds where only one agent is muddy, agent $a$ establishes the correspondence function~$\C$ to connect each world from $\calM^a$ to a  propositionally equivalent world of $\calM^{\neg a}$ if such a world exists. It is easy to see that the coherency conditions for the a~priori belief update are fulfilled.
According to Def.~\ref{def:APBU}, the result $(\calM\circledcirc_a\U ,\underline{A})$ of applying the update to the initial model  is shown in \cref{fig:EX1.3}.
The two arrows from $A \in U^a$ for $b$ and $c$ to worlds in $W^{\neg a}$ are added because of the correspondence relation~$\C$. For instance, the $b$-arrow from $A$ in $U^a$ to $AB$ in $\calM^{\neg a}$ of the updated model is due to the $b$-arrow from $A$ to $AB$ in $\calM^a$ and the correspondence from $AB$ in $\calM^a$ to $AB$ in $\calM^{\neg a}$.
As a result 
\begin{align*}
\calM\circledcirc_a\U ,\underline{A}  &\models B_b (m_a \land  m_b \land \neg m_c),
&
\calM\circledcirc_a\U ,\underline{A}  &\models B_aB_b (m_a \land  m_b \land \neg m_c),
\\
\calM\circledcirc_a\U ,\underline{A}  &\models B_c (m_a \land \neg m_b \land  m_c),
&
\calM\circledcirc_a\U ,\underline{A}  &\models B_aB_c (m_a \land \neg m_b \land  m_c).
\\
\calM\circledcirc_a\U ,\underline{A} &\models B_a (m_a \land \neg m_b \land \neg m_c),
&
\calM\circledcirc_a\U ,\underline{A}  &\models B_bB_a (m_a \land  m_b \land  \neg m_c),
\\
&&\calM\circledcirc_a\U ,\underline{A}  &\models B_cB_a (m_a \land  \neg m_b \land  m_c),
\end{align*}
In other words, each agent believes to know the actual situation and to be muddy. All three of them will step forward when Father prompts. Moreover, the a~priori update restored $a$'s ability to understand the reasoning of others. Because $a$ guessed $\calM^{\neg a}$ correctly, she can correctly interpret $b$ and $c$ stepping forward. Agents~$b$~and~$c$ expect $a$ to step forward. On the other hand, $b$~cannot interpret $c$~stepping forward  because 
$
\calM\circledcirc_a\U ,\underline{A} \models B_b (\neg B_c m_c \land \neg B_c \neg m_c)$ and vice versa.

\end{example}

\subsection{A~priori Belief Update Triggered by  Public Announcements}
\label{subsec:example_PA}

One could say that \cref{MCEX:2} should have been modeled according to \cref{fig:EX1.3} from the very beginning because it fits better with the evidence observed by the agents than model~$\calM^{\neg a}$ from \cref{fig:EX1.1} they attempt to use. That would not address  the question of how to turn an epistemic description of agents' observations and beliefs in complex epistemic scenarios, possibly with  false common beliefs of a subset of agents. Our a~priori belief updates provide a mechanism for building models for more complex epistemic situations based on existing models for simpler scenarios.

It is harder to find an alternative to a~priori belief updates when the inadequacy of the a~priori assumptions cannot be observed initially and is only uncovered through communication. As already mentioned, the idea of scrapping the original model and starting anew corresponds to the development cycle of ordinary distributed systems, where it is performed by a system designer. We aim at developing a mechanism for individual agents to do it, as befits SASO systems. Moreover, as we saw in  \cref{MCEX:2}, different agents are likely to discover the flaw in their assumptions at different times, which puts the idea of scrapping the \emph{whole} model at odds with the Knowledge of Preconditions Principle~\cite{Mos15TARK} that requires an agent performing an action to know the reason for this action.

This was the situation in \cref{ex:consec} where it was $a$'s public announcement of knowing $b$'s number that caused $b$ to realize that some a~priori assumptions must have been wrong. We now show how to perform a~priori belief update in such a situation by representing $b$'s  reasoning in \cref{ex:consec} in our formal framework.

\begin{example}[Consecutive numbers with false a~priori assumptions formalized]
\label{ex:consec2}
Let us describe the twist version of \cref{ex:consec}  as a Kripke model. 
As before, we label  worlds according to pairs of numbers held by the agents. 
Although this  results in several worlds having  the same label, we prefer to disambiguate when necessary rather than creating an overly complex nomenclature for worlds. 
As before, the actual world~$\underline{(1,2)}$ is distinguished by the underline. 
Unlike the MCP, an important a priori assumption is not common between agents $a$ and $b$: $a$ starts natural numbers from $1$ while $b$ also considers $0$ to be natural. 
At the same time, none of them are initially aware of this disagreement. 
Hence, pointed Kripke model $\left(\calN,\underline{(1,2)}\right)$ in \cref{fig:EX4.1}~(Above) is introspective but not epistemic and consists of the actual world and two disjoint fragments, each only accessible to one of the two agents: the upper line of worlds represents what $a$ (mistakenly) thinks is their common view of the situation, while the lower line of worlds represents what $b$ (mistakenly) takes to be their common view. 
To distinguish similarly labeled worlds, we call the worlds from the upper~(lower)~line  $a$-worlds ($b$-worlds). To simplify notation without affecting the logical content of the puzzle, suppose $a$ states $b$'s number explicitly by saying ``I know that you have number $2$,'' which we represent by formula $B_a 2_b$. 
It is easy to see that $\calN, w \notmodels B_a 2_b$ for any $b$-world~$w$; of the $a$-worlds, only  world~$(1,2)$ satisfies $B_a 2_b$, and so does the actual world.
Accordingly, the result of a public update with~$B_a 2_b$, as  depicted in \cref{fig:EX4.1} (Below),  contains only two worlds and $\calN\upd B_a 2_b, \underline{(1,2)} \models B_b \bot$, prompting the ``Wait, what?'' comment in our informal rendition and triggering an a~priori update by $b$.

\begin{figure}
\begin{subfigure}{.99\textwidth}
 \centering
  \resizebox{8cm}{!}{
    \begin{tikzpicture}
\node[] at (0,0) [label=135:\textbf{\large{$\left(\calN, \underline{(1,2)}\right)$}}] (@) {$\underline{(1,2)}$};

\node (A) at (4,1) {$(1,2)$};
\node (A1) [right=2cm of A] {$(3,2)$};
\node (A2) [right=2cm of A1] {$(3,4)$};
\node (A3) [right=0.5cm of A2] {$\ldots$};

\node (B) at (2,-1.5) {$(1,0)$};
\node (B1) at (4,-1.5) {$(1,2)$};
\node (B2) [right=2cm of B1] {$(3,2)$};
\node (B3) [right=2cm of B2] {$(3,4)$};
\node (B4) [right=0.5cm of B3] {$\ldots$};

\node[] at (0,-6) [label=135:\textbf{\large{$\left(\calN\upd B_a 2_b, \underline{(1,2)}\right)$}}] (@') {$\underline{(1,2)}$};

\node (A') at (4,-5) {$(1,2)$};
\node (A3') [right=4.5cm of A'] {};

\path[->, -{Stealth[scale width=1.5]}] 
(@) edge node[above] {$a$} (A)
(@) edge node[below] {$b$} (B1)
(@) edge node[above] {$b$} (B2);

\path[<->,{Stealth[scale width=1.5]}-{Stealth[scale width=1.5]}] 
(A) edge node[above] {$b$} (A1)
(A1) edge node[above] {$a$} (A2)
(B) edge node[below] {$a$} (B1)
(B1) edge node[below] {$b$} (B2)
(B2) edge node[below] {$a$} (B3)
;

\path[->,-{Stealth[scale width=1.5]}] 
          (A) edge[loop above] node[above] {$a,b$} (A)
          (A1) edge[loop above] node[above] {$a,b$} (A1)
          (A2) edge[loop above] node[above] {$a,b$} (A2)
          (B) edge[loop below] node[below] {$a,b$} (B)
          (B1) edge[loop below] node[below] {$a,b$} (B1)
          (B2) edge[loop below] node[below] {$a,b$} (B2)
          (B3) edge[loop below] node[below] {$a,b$} (B3)
          ;

\path[->, -{Stealth[scale width=1.5]}] 
(@') edge node[above] {$a$} (A');

\path[->,-{Stealth[scale width=1.5]}] 
          (A') edge[loop above] node[above] {$a,b$} (A')
          ;

\end{tikzpicture}
}
    \caption{Above: Introspective pointed Kripke model $\left(\calN, \underline{(1,2)}\right)$ representing the initial state of both agents in the consecutive numbers puzzle from \cref{ex:consec2}, where $b$ considers $0$ to be a natural number while $a$ does not.\\
    Below: $\left(\calN\upd B_a 2_b, \underline{(1,2)}\right)$ after  $a$'s public announcement that $a$ knows $b$'s number is 2.}
    \label{fig:EX4.1}
\end{subfigure}

\par\bigskip 

\begin{subfigure}{.99\textwidth}
\centering
\resizebox{7.5cm}{!}{
    \begin{tikzpicture}

\node (B) at (0,0) [label=135:\textbf{\large{$\calM^b$}}] {$(1,0)$};
\node (B1) [right=2cm of B, label=135:\textbf{\large{$U^b$}}] {$(1,2)$};
\node (B2) [right=2cm of B1] {$(3,2)$};
\node (B3) [right=2cm of B2] {$(3,4)$};
\node (B4) [right=0.5cm of B3] {$\ldots$};

\node (A) [below=2cm of B1, label=135:\textbf{\large{$\calM^{\neg b}$}}] {$(1,2)$};
\node (A1) [right=2cm of A] {$(3,2)$};
\node (A2) [right=2cm of A1] {$(3,4)$};
\node (A3) [right=0.5cm of A2] {$\ldots$};

\path[<->,{Stealth[scale width=1.5]}-{Stealth[scale width=1.5]}] 
(A) edge node[above] {$b$} (A1)
(A1) edge node[above] {$a$} (A2)
(B) edge node[below] {$a$} (B1)
(B1) edge node[below] {$b$} (B2)
(B2) edge node[below] {$a$} (B3)
;

\path[dashed] 
(B1) edge node[] {} (A)
(B2) edge node[] {} (A1)
(B3) edge node[] {} (A2)
;

\path[->,-{Stealth[scale width=1.5]}] 
          (A) edge[loop below] node[below] {$a,b$} (A)
          (A1) edge[loop below] node[below] {$a,b$} (A1)
          (A2) edge[loop below] node[below] {$a,b$} (A2)
          (B) edge[loop above] node[above] {$a,b$} (B)
          (B1) edge[loop above] node[above] {$a,b$} (B1)
          (B2) edge[loop above] node[above] {$a,b$} (B2)
          (B3) edge[loop above] node[above] {$a,b$} (B3)
          ;

\begin{pgfonlayer}{background}
\filldraw [line width=16mm,line join=round,black!10]
      (B1.north  -| B2.east)  rectangle (B1.south  -| B1.west);

\end{pgfonlayer}
\end{tikzpicture}
}
    \caption{Elements of the a~priori belief update $\U = (\calM^b, U^b, \calM^{\neg b}, \C)$ for $b$ in \cref{ex:consec2}. The correspondence relation is represented by dashed lines. The grey rectangle is $b$'s equivalence class $U^b$.}
    \label{fig:EX4.2}
\end{subfigure}

\par\bigskip 

\begin{subfigure}{.99\textwidth}
    \resizebox{12cm}{!}{
    \begin{tikzpicture}

\node (@') at (-6.3,3.5) [label=135:\textbf{\large{$\left( (\calN\upd B_a 2_b) \circledcirc_b \U, \underline{(1,2)}\right)$}}] {};

\node (C') at (-10, 2) [label=135:\textbf{\large{$W$}}] {$\underline{(1,2)}$};
\node (C1') [right=2cm of C'] {$(1,2)$};

\node (B1') [below=2cm of C1', label=90:\textbf{\large{$U^b$}}] {$(1,2)$};
\node (B2') [right=2cm of B1'] {$(3,2)$};

\node (A') [below=2cm of B1', label=135:\textbf{\large{$W^{\neg b}$}}] {$(1,2)$};
\node (A1') [right=2cm of A'] {$(3,2)$};
\node (A2') [right=2cm of A1'] {$(3,4)$};
\node (A3') [right=0.5cm of A2'] {$\ldots$};

\node (@) [right=14.5cm of @', label=135:\textbf{\large{$\left(\bigl((\calN\upd B_a 2_b)  \circledcirc_b \U \bigr) \updpriv{b} B_a 2_b, \underline{(1,2)}\right)$}}] {};
\node (C) [right=14cm of C', label=135:\textbf{\large{$W$}}] {$\underline{(1,2)}$};
\node (C1) [right=2cm of C] {$(1,2)$};
\node (B1) [below=2cm of C1, label=45:\textbf{\large{$U^b\upd B_a 2_b$}}] {$(1,2)$};
\node (A) [below=2cm of B1, label=135:\textbf{\large{$W^{\neg b}\upd B_a 2_b$}}] {$(1,2)$};

\node (L) [left=6cm of @] {};
\node (L1) [below=8cm of L] {};

\path[thick]
(L) edge node[] {} (L1);

\path[<->,{Stealth[scale width=1.5]}-{Stealth[scale width=1.5]}] 
(A') edge node[above] {$b$} (A1')
(A1') edge node[above] {$a$} (A2')
(B1') edge node[below] {$b$} (B2')
;

\path[->,-{Stealth[scale width=1.5]}] 
(C') edge node[above] {$b$} (B1')
(C') edge node[above, pos=.6] {$b$} (B2')
(C') edge node[above] {$a$} (C1')
(B1') edge node[right] {$a$} (A')
(B2') edge node[right] {$a$} (A1')
(B2') edge node[right] {$a$} (A2')
;

\path[->,-{Stealth[scale width=1.5]}] 
(C) edge node[above] {$b$} (B1)
(C) edge node[above] {$a$} (C1)
(B1) edge node[right] {$a$} (A);

\path[->,-{Stealth[scale width=1.5]}] 
          (C1') edge[loop above] node[above] {$a,b$} (C1')
          (A') edge[loop below] node[below] {$a,b$} (A')
          (A1') edge[loop below] node[below] {$a,b$} (A1')
          (A2') edge[loop below] node[below] {$a,b$} (A2')
          (B1') edge[loop left] node[left] {$b$} (B1')
          (B2') edge[loop right] node[right] {$b$} (B2')
;

\path[->,-{Stealth[scale width=1.5]}] 
          (C1) edge[loop above] node[above] {$a,b$} (C1)
          (A) edge[loop below] node[below] {$a,b$} (A)
          (B1) edge[loop left] node[left] {$b$} (B1)
          ;

\begin{pgfonlayer}{background}
\filldraw [line width=16mm,line join=round,black!10]
      (B1'.north  -| B2'.east)  rectangle (B1'.south  -| B1'.west);
      \filldraw [line width=16mm,line join=round,black!10]
      (C'.north  -| C1'.east)  rectangle (C'.south  -| C'.west);
      \filldraw [line width=16mm,line join=round,black!10]
      (A'.north  -| A3'.east)  rectangle (A'.south  -| A'.west);
      \filldraw [line width=16mm,line join=round,black!10]
      (B1.north  -| B1.east)  rectangle (B1.south  -| B1.west);
      \filldraw [line width=16mm,line join=round,black!10]
      (C.north  -| C1.east)  rectangle (C.south  -| C.west);
      \filldraw [line width=16mm,line join=round,black!10]
      (A.north  -| A.east)  rectangle (A.south  -| A.west);

\end{pgfonlayer}
\end{tikzpicture}
}
    \caption{Left: Result of $b$ applying a~priori belief update in the inconsistent epistemic state $\left(\calN\upd B_a 2_b, \underline{(1,2)}\right)$ of \cref{fig:EX4.1}.\\ 
    Right: Result of $b$ (re)applying $a$'s public announcement of $B_a 2_b$ to $b$'s part of the model.}
\label{fig:EX4.3}
\end{subfigure}
\caption{}
\end{figure}

To match the informal reasoning from \cref{ex:consec}, consider the nature of this public announcement: it was about $a$'s beliefs rather than propositional facts. Hence, it is rational for $b$ to return to the original model she considered before the public announcement: model~$\calM^b$ from \cref{fig:EX4.2} is nothing but $b$-worlds from \cref{fig:EX4.1} (Above) and $U^b$ is the $b$-cluster of worlds $(1,2)$ and $(3,2)$ that $b$ considered possible pre-announcement. What $b$ should modify is how $a$ sees the situation. For that $b$ correctly imagines the abbreviated natural numbers, resulting in $\calM^{\neg b}$ being the same as $a$-worlds from \cref{fig:EX4.1} (Above). Once again, the correspondence is determined by propositional truth of atoms $n_a$ and $m_b$. It is easy to see that all coherency conditions are satisfied. 

The result of $b$ applying a~priori update $\U$ to $\left(\calN\upd B_a 2_b,\underline{(1,2)}\right)$ is shown in \cref{fig:EX4.3} (Left). But this is not the model explaining why $b$ now knows $a$'s number. Indeed, this cannot be the final stage of $b$'s reevaluation since $b$'s part of this model does not reflect the public announcement $a$ made. Agent $b$ still needs to apply the standard public announcement update to the only part of the model it is aware of, which is comprised of worlds from $U^b$ and $W^{\neg b}$. This leaves $b$'s part of the model with two worlds only, $(1,2)$ from $U^b$ and $(1,2)$ from $W^{\neg b}$, as shown in \cref{fig:EX4.3} (Right). It can be easily seen that $\left((\calN\upd B_a 2_b)\circledcirc_b \U\right) \updpriv{b} B_a 2_b, \underline{(1,2)} \models B_b 1_a$, explaining that $b$ has now figured out $a$'s number too.

Such a ``private'' public announcement update $\updpriv{b}$ may not be standard, but it is warranted in this setting. Recall that $b$ here assumes the role of the system designer and treats her part of the model as the whole model and the only model pre-announcement. It is natural then for $b$ to update this partial model as if it is the whole model. This operation should not affect $a$'s part of the model since $b$ is not aware of it.
\end{example}

If, in the a~priori updated model after the public announcement there is at least one world in $U^a$ that makes the announced formula true, then the self-recovery operation is indeed successful, as stated by the following corollary:

\begin{corollary}
\label{cor:self-recov}
    Let $\U=(\calM^a, U^a, \calM^{\neg a}, \C)$ be an a~priori belief update triggered by the public announcement of $\phi$
and  $(\calM,v)$ be a pointed Kripke model with introspective model $\calM = \langle W, R, V\rangle$ pre-an\-nounce\-ment. 
If $(\calM\upd\phi) \circledcirc_a \U, u \models \phi$ for some world $u \in U^a$, then $a$'s beliefs are consistent after the a~priori update, i.e., $\bigl((\calM \upd \phi) \circledcirc_a \U\bigr) \updpriv{a} \varphi, v  \notmodels B_a \bot$.
\end{corollary}
\begin{proof}
    The statement follows from~\cref{Th:cons_Bel}.\ref{consist}.
\end{proof}
Note that because of Moorean sentences, there is generally no guarantee that $a$ believes in $\varphi$ after the public announcement.

\subsection{Simultaneous A~priori Belief Updates by Several Agents}
\label{subsec:example_Simultaneous}

So far we considered only scenarios where beliefs of exactly one agent become inconsistent, resulting in a~priori belief update for this agent. It is, of course, entirely possible that several agents become inconsistent simultaneously. Given that an a~priori belief update is performed by an agent privately and has no effect on the rest of the model, it is straightforward to apply several such updates in parallel.

\begin{figure}
\begin{subfigure}{.99\textwidth}
 \centering
  \resizebox{7.5cm}{!}{
    \begin{tikzpicture}
\node[nod, fill=gray, label=135:{\textbf{\Large{$W\upd\ASF$}}}] at (-6,0) [label=below:$\underline{A}$] (@) {};
\node (p) at (-8, 1.6 ) {\textbf{\Large{$\Bigl(\bigl(\calM\circledcirc_a\U\bigr) \upd \ASF ,\underline{A}\Bigr)$}}};
\node[nod, label=135:{\textbf{\large{$U^a\upd\ASF$}}}] at (0,0) [label=below:$A$] (A) {};
\node (110) [right=6cm of A] {\textbf{\Large{$W^{\neg a}\upd\ASF$}}};
\path[->, -{Stealth[scale width=1.5]}] (@)
edge node[below] {$a$} (A);
\path[->,-{Stealth[scale width=1.5]}] 
          (A) edge[loop above] node[above] {$a$} (A);

\begin{pgfonlayer}{background}
\filldraw [line width=15mm,line join=round,black!10]
      (A.north  -| A.east)  rectangle (A.south  -| A.west);
\filldraw [line width=10mm,line join=round,black!10]
      (110.north  -| 110.east)  rectangle (110.south  -| 110.west);
\filldraw [line width=15mm,line join=round,black!10]
      (@.north  -| @.east)  rectangle (@.south  -| @.west);
\end{pgfonlayer}
\end{tikzpicture}
}
    \caption{Pointed Kripke model representing the result of the public announcement of $\ASF$ on the a~priori updated model~\mbox{$(\calM\circledcirc_a\U,\underline{A})$} of \cref{MCEX:2}, as  described in \cref{ex:simul}.}
    \label{fig:EX1.4}
\end{subfigure}

\par\bigskip 

\begin{subfigure}{.99\textwidth}
 \resizebox{10.5cm}{!}{
\begin{tikzpicture}
\node[nod, fill=gray, label=135:{\textbf{\Large{$W\upd\ASF$}}}] at (-7,0) [label=below:$\underline{A}$] (@) {};

\node[nod] at (-2,-2) [label=220:$AB$, label=135:$\large{U^b}$] (B) {};
\node[nod] [right=2cm of B, label=below:$A$] (B1) {};

\node[nod] [below=3.5cm of B1, label=45:$AC$, label=180:$\large{U^c}$] (C) {};
\node[nod] [below=1.5cm of C, label=below:$A$] (C1) {};

\node[nod] at (6,-1) [label=below:$AB$] (110) {};
\node[nod] at (6,1) [label=above:$ABC$] (111) {};
\node[nod] at (9,1) [label=above:$AC$] (101) {};
\node[nod] at (8,3) [label=above:$BC$] (011) {};

\node[nod] at (6,-8) [label=below:$AB$] (110') {};
\node[nod] at (6,-6) [label=above:$ABC$] (111') {};
\node[nod] at (9,-6) [label=above:$AC$] (101') {};
\node[nod] at (8,-4) [label=above:$BC$] (011') {};

\node (n) [above=5cm of 111] {};
\node (w) [right=.01cm of n] {$W^{\neg a}\upd\ASF$};
\node (w1) [above=2cm of 111] {$W^{\neg b}$};
\node (w2) [above=2cm of 111'] {$W^{\neg c}$};

\node (p) [above=4cm of @] {\textbf{\Large{$\Bigl( \bigl((\calM\circledcirc_a\U ) \mid  \ASF \bigr) \circledcirc_{b,c} \U',\underline{A}\Bigr)$}}};

\node[nod] [left=5cm of w, label=below:$A$, label=135:$\large{U^a}\upd\ASF$] (A) {};

\path[<->, {Stealth[scale width=1.5]}-{Stealth[scale width=1.5]}]
(111)
edge node[above] {$b$} (101)
edge node[above] {$a$} (011)
edge node[right, pos=.4] {$c$} (110)
(111')
edge node[above] {$b$} (101')
edge node[above] {$a$} (011')
edge node[right, pos=.4] {$c$} (110')
(B)
edge node[below] {$b$} (B1)
(C)
edge node[right] {$c$} (C1)
;

\path[->, -{Stealth[scale width=1.5]}] 
(@)
edge[] node[above] {$a$} (A)
edge[] node[below] {$b$} (B)
edge[] node[above] {$c$} (C)
edge[] node[above] {$b$} (B1)
edge[] node[above] {$c$} (C1)
(B)
edge[] node[below] {$c$} (111)
edge[bend right=30] node[below] {$c$} (110)
(B1)
edge[] node[below] {$c$} (101)
(C)
edge[bend left=20] node[above] {$b$} (101')
(C)
edge[bend right=15] node[above] {$b$} (111')
(C1)
edge[] node[below] {$b$} (110')
;

\path[->,-{Stealth[scale width=1.5]}] 
          (011) edge[loop left] node[left] {$a,b,c$} (011)
          (111) edge[loop left] node[left] {$a,b,c$} (111)
          (110) edge[loop right] node[right] {$a,b,c$} (110)
          (101) edge[loop below] node[below] {$a,b,c$} (101)
          (011') edge[loop left] node[left] {$a,b,c$} (011')
          (111') edge[loop left] node[left] {$a,b,c$} (111')
          (110') edge[loop right] node[right] {$a,b,c$} (110')
          (101') edge[loop below] node[below] {$a,b,c$} (101')
          (A) edge[loop above] node[above] {$a$} (A)
          (B) edge[loop below] node[below] {$b$} (B)
          (C) edge[loop above] node[above] {$c$} (C)
          (B1) edge[loop right] node[right] {$b$} (B1)
          (C1) edge[loop left] node[left] {$c$} (C1)
          ;

\begin{pgfonlayer}{background}

\filldraw [line width=14mm,line join=round,black!10]
      (A.north  -| A.east)  rectangle (A.south  -| A.west);

\filldraw [line width=14mm,line join=round,black!10]
      (B.north  -| B1.east)  rectangle (B.south  -| B.west);

\filldraw [line width=14mm,line join=round,black!10]
      (C.north  -| C.east)  rectangle (C1.south  -| C.west);

\filldraw [line width=12mm,line join=round,black!10]
      (w.north  -| w.east)  rectangle (w.south  -| w.west);
      
\filldraw [line width=12mm,line join=round,black!10]
      (011.north  -| 101.east)  rectangle (110.south  -| 110.west);

\filldraw [line width=12mm,line join=round,black!10]
      (011'.north  -| 101'.east)  rectangle (110'.south  -| 110'.west);

\filldraw [line width=12mm,line join=round,black!10]
      (@.north  -| @.east)  rectangle (@.south  -| @.west);
\end{pgfonlayer}

\end{tikzpicture}
}
    \caption{Pointed Kripke model showing the intermediate state of \cref{ex:simul} after $b$ and $c$ independently performed  a~priori belief updates in parallel but have not yet (re)applied the public announcement of~$\ASF$.}
    \label{fig:EX1.5}
\end{subfigure}

\par\bigskip 

\begin{subfigure}{.99\textwidth}
    \resizebox{10cm}{!}{
    \begin{tikzpicture}
\node[nod, label=135:{\textbf{\large{$U^a\upd\ASF$}}}] at (5.5,0) [label=below:$A$] (A) {};
\node (110) [right=5cm of A] {\textbf{\Large{$W^{\neg a}\upd\ASF$}}};

\node[nod, fill=gray, label=135:{\textbf{\Large{$W\upd\ASF$}}}] at (-2,-3) [label=below:$\underline{A}$] (@') {};

\node at (-2.7,-1)  [] {\textbf{\large{$\Bigl(\left(\bigl((\calM\circledcirc_a\U ) \mid  \ASF \bigr) \circledcirc_{b,c} \U'\right)\updpriv{b,c}\ASF,\underline{A}\Bigr)$}}};

\node[nod, label=135:{\textbf{\large{$U^b\upd\ASF$}}}] at (5.5,-3) [label=below:$A$] (A') {};

\node (110') [right=5cm of A'] {\textbf{\Large{$W^{\neg b}\upd\ASF$}}};

\node[nod, label=180:{\textbf{\large{$U^c\upd\ASF$}}}] at (5.5,-6) [label=below:$A$] (A'') {};
\node (110'') [right=5cm of A''] {\textbf{\Large{$W^{\neg c}\upd\ASF$}}};

\path[->, -{Stealth[scale width=1.5]}] 
(@')
edge node[below] {$a$} (A) 
(A) edge[loop above] node[above] {$a$} (A)
(@')
edge node[below] {$b$} (A')
(A') edge[loop above] node[above] {$b$} (A')
(@')
edge node[below] {$c$} (A'')
(A'') edge[loop above] node[above] {$c$} (A'');

\begin{pgfonlayer}{background}
\filldraw [line width=15mm,line join=round,black!10]
      (A.north  -| A.east)  rectangle (A.south  -| A.west);
\filldraw [line width=10mm,line join=round,black!10]
      (110.north  -| 110.east)  rectangle (110.south  -| 110.west);

\filldraw [line width=15mm,line join=round,black!10]
      (A'.north  -| A'.east)  rectangle (A'.south  -| A'.west);
\filldraw [line width=10mm,line join=round,black!10]
      (110'.north  -| 110'.east)  rectangle (110'.south  -| 110'.west);
\filldraw [line width=15mm,line join=round,black!10]
      (@'.north  -| @'.east)  rectangle (@'.south  -| @'.west);

\filldraw [line width=15mm,line join=round,black!10]
      (A''.north  -| A''.east)  rectangle (A''.south  -| A''.west);
\filldraw [line width=10mm,line join=round,black!10]
      (110''.north  -| 110''.east)  rectangle (110''.south  -| 110''.west);
\end{pgfonlayer}
\end{tikzpicture}
}
    \caption{Pointed Kripke model illustrating the final stage of \cref{ex:simul} after a~priori belief updates of $b$ and $c$ and their application of the public announcement of $\ASF$ to their newly created parts of the model.}
    \label{fig:EX1.6}
\end{subfigure}
\caption{}
\end{figure}

\begin{example}[Simultaneous a~priori belief updates triggered by a public announcement]
\label{ex:simul}
Let us continue the scenario from \cref{MCEX:2} where we left off, i.e., after $a$ has restored her consistency, Father asked the children if they know whether they are muddy, and all three children stepped forward, which is equivalent to the public announcement of ``all step forward''
\begin{equation}
\label{eq:asf}
\ASF = (B_a m_a \lor B_a \neg m_a) \land (B_b m_b \lor B_b \neg m_b) \land (B_c m_c \lor B_c \neg m_c).
\end{equation}
The resulting model $\calM'\ce\bigl(\calM\circledcirc_a\U\bigr) \upd \ASF$ is shown in \cref{fig:EX1.4} and has only two worlds remaining from \cref{fig:EX1.3}, the actual world $\underline{A}$ and world $A$ from $U^a$. Thus, beliefs of~$b$~and~$c$ have become inconsistent (there~are no outgoing $b$- or $c$-arrows from the actual world $\underline{A}$) and $a$ knows that this is the case (there are no outgoing $b$- or $c$-arrows from $A$, the only world $a$ considers possible). At the same time, $a$ still maintains consistency of beliefs and thinks of herself as muddy:
\[
\calM' ,\underline{A} \models B_b \bot \land B_c \bot \land B_aB_b \bot \land B_a  B_c \bot\land \neg B_a \bot \land B_a m_a.
\]

At this point $b$ and $c$ independently perform each their own a~priori belief update triggered by the announcement of $\ASF$. Although $a$ can expect them to do so, one could argue that, due to the non-deterministic \emph{ad~hoc} nature of a~priori belief updates, $a$~does not think she can guess how $b$ and $c$ would choose to update their beliefs, nor is such a guess necessary for $a$ to be able to respond to Father. Accordingly, $a$ remains satisfied in her belief that she is muddy and that her companions are confused.

Suppose that each of~$b$~and~$c$ chooses to do the same thing that $a$ has done earlier, i.e., using $\calM^a$ from \cref{fig:EX1.1} (Middle) as both $\calM^b$ and $\calM^c$ and using $\calM^{\neg a}$  from \cref{fig:EX1.1} (Right) as both $\calM^{\neg b}$ and $\calM^{\neg c}$. Correspondence $\C$ is also the same for both~$b$~and~$c$. However, their local states $U^b = \{A, AB\}$ and $U^c = \{A, AC\}$ differ due to their differing points of view.

We present the state after $b$ and  $c$'s simultaneous independent a~priori update in \cref{fig:EX1.5}. They then each apply $\ASF$ to their private parts of the model respectively. No worlds in $W^{\neg b}$/$W^{\neg c}$ satisfy $\ASF$, so all are pruned. World~$AB$ of $U^b$ is rejected because there $c$ does not know its state, 
\[
\bigl((\calM\circledcirc_a\U ) \mid  \ASF \bigr) \circledcirc_{b,c} \U', AB \notmodels B_c m_c \lor B_c \neg m_c;
\]
similarly $AC$ of $U^c$ does not survive the announcement of $\ASF$ since $b$ there is not sure of its state. This leaves only worlds $A$ from $U^b$ and $U^c$ (and the whole $a$'s part of the model, which is not affected by a~priori thinking of $b$ and $c$), resulting in \cref{fig:EX1.6}.
Note that now each agent believes that the beliefs of other two agents are inconsistent while themselves correctly believing that $a$ is the only muddy child. In effect, all children solved the problem correctly but lost  trust in each other's reasoning: using 
\[
\calM''\ce\Bigl(\bigl((\calM\circledcirc_a\U ) \mid  \ASF \bigr) \circledcirc_{b,c} \U'\Bigr)\updpriv{b,c}\ASF
\]
for the model, $E\phi \ce B_a \phi \land B_b \phi \land B_c \phi$ for mutual knowledge, and $\hat{E}\phi  \ce \neg E \neg \phi$ for its dual,
\[
\calM'',\underline{A} \models E (m_a \land \neg m_b \land \neg m_c) \land \neg \hat{E} \bot \land B_a(B_b \bot \land B_c\bot) \land B_b(B_a \bot \land B_c \bot) \land B_c (B_a \bot \land B_b \bot).
\]

Note that update $\U'$ in~\cref{fig:EX1.5} represents two separate a~priori belief updates, one by~$b$ (rectangles $U^b$ and $W^{\neg b}$) and another by $c$ (rectangles $U^c$ and $W^{\neg c}$). Hence, here we take $\U'$ to be a partial function from the set $\agents$ of agents to a~priori belief update tuples, so that $\U'(b)$ is the a~priori update performed by $b$ while $\U'(c)$ is the one for $c$. This notation is similar to using $R$ for accessibility relations of all agents. The domain of $\U'$ is then listed in the subscript to $\circledcirc$ as in $ \bigl((\calM\circledcirc_a\U ) \mid  \ASF \bigr) \circledcirc_{b,c} \U'$. This does not cause any formal problems because,
 being private, individual a~priori belief updates  do not interfere with each other.

\end{example}

\subsection{A~priori Belief Updates Need Not Yield (Correct) Solutions}
\label{subsec:example_unsuccess}

In all examples considered so far, all agents performing the updates have succeeded in recovering a consistent epistemic state (though sometimes at the expense of expecting consistency from other agents). By no means do we claim that this is always the case. The \emph{ad~hoc} chosen a~priori belief update $\U$ need not lead to the resolution of the agent's conundrum. It is rather a matter of luck whether it would. For instance, consider the setup of  \cref{ex:consec2}. Previously, we had agent $b$ guess the exact model of the number line used by~$a$. But~$b$'s~guess can also be wrong. For instance, $b$ might think that $a$~mistakenly considers all integers rather than non-negative integers only, making $\calM^{\neg b}$ from \cref{fig:EX4.3} extend infinitely in both directions, retaining world~$(1,0)$, and not resulting in a consistent state for $b$ after  $B_a 2_b$ is taken into account. If the agent is persistent, it should be expected that such an unsuccessful update is rejected and a new attempt is made to restore consistency using a different a~priori belief update.

It is also possible that a wrong a~priori guess does  not result in an inconsistency but leads to wrong conclusions, undetected by any of the agents.
\begin{example}[Simultaneous a~priori belief updates creating consistent false beliefs]
\label{MCEX:4}
Consider a variant of \cref{MCEX:2} where all agents have the same initial explicit a~priori common belief $\APB_2$, but all are clean in actuality.
In this case, all agents detect inconsistency from the start, and each privately and independently performs an a~priori belief update. Suppose that, using the same reasoning as agent $a$ in \cref{MCEX:2} and agents $b$ and $c$ in \cref{ex:simul}, agents use a simultaneous update $\U''$ such that:  
\[
\U''_a \ce (\calM^a,\{A\}, \calM^{\neg a}, \C),
\qquad
\U''_b \ce (\calM^b,\{B\}, \calM^{\neg b}, \C),
\qquad
\U''_c \ce (\calM^c,\{C\}, \calM^{\neg c}, \C),
\]
where the only difference to the previous case is in their local states $V^a= \{A\}$, $V^b=\{B\}$, and $V^c=\{C\}$ that are chosen based on what they observe. The resulting model  is depicted in \cref{fig:EX3.1}. It is easy to see that:\looseness=-1
\[
\calM \circledcirc_{a,b,c} \U'',\underline{0} \models B_a (m_a \land B_b m_b \land B_c m_c)  \land B_b (m_b \land B_a m_a \land B_c m_c) \land B_c (m_c \land B_a m_a \land B_b m_b).
\]
\begin{figure}
\begin{subfigure}{.99\textwidth}
 \centering
  \resizebox{10cm}{!}{
\begin{tikzpicture}
\node[nod, fill=gray, label=135:{\textbf{\Large{$W$}}}] at (-7,0) [label=below:$\underline{0}$] (@) {};

\node[nod] at (0,-2) [label=below:$B$, label=135:$\large{V^b}$] (B) {};

\node[nod] [below=3.5cm of B, label=45:$C$, label=180:$\large{V^c}$] (C) {};

\node[nod] [above=7.5cm of B, label=below:$A$, label=135:$\large{V^a}$] (A) {}; 

\node[nod] at (6,5) [label=below:$AB$] (110'') {};
\node[nod] at (6,7) [label=above:$ABC$] (111'') {};
\node[nod] at (9,7) [label=above:$AC$] (101'') {};
\node[nod] at (8,9) [label=above:$BC$] (011'') {};

\node[nod] at (6,-1) [label=below:$AB$] (110) {};
\node[nod] at (6,1) [label=above:$ABC$] (111) {};
\node[nod] at (9,1) [label=above:$AC$] (101) {};
\node[nod] at (8,3) [label=above:$BC$] (011) {};

\node[nod] at (6,-8) [label=below:$AB$] (110') {};
\node[nod] at (6,-6) [label=above:$ABC$] (111') {};
\node[nod] at (9,-6) [label=above:$AC$] (101') {};
\node[nod] at (8,-4) [label=above:$BC$] (011') {};

\node (w) [above=2cm of 111''] {$W^{\neg a}$};
\node (w1) [above=2cm of 111] {$W^{\neg b}$};
\node (w2) [above=2cm of 111'] {$W^{\neg c}$};

\node (p) [above=4cm of @] {\textbf{\Large{$( \calM \circledcirc_{a,b,c} \U'',\underline{0})$}}};

\path[<->, {Stealth[scale width=1.5]}-{Stealth[scale width=1.5]}]
(111)
edge node[above] {$b$} (101)
edge node[above] {$a$} (011)
edge node[right, pos=.4] {$c$} (110)
(111')
edge node[above] {$b$} (101')
edge node[above] {$a$} (011')
edge node[right, pos=.4] {$c$} (110')
(111'')
edge node[above] {$b$} (101'')
edge node[above] {$a$} (011'')
edge node[right, pos=.4] {$c$} (110'')
;

\path[->, -{Stealth[scale width=1.5]}] 
(@)
edge[] node[above] {$a$} (A)
edge[] node[above] {$b$} (B)
edge[] node[above] {$c$} (C)
(A)
edge[] node[above] {$b$} (110'')
edge[] node[above] {$c$} (101'')
(B)
edge[bend left=15] node[above] {$c$} (011)
edge[] node[above] {$a$} (110)
(C)
edge[] node[above] {$b$} (011')
edge[bend right=15] node[above] {$a$} (101')
;

\path[->,-{Stealth[scale width=1.5]}] 
          (011) edge[loop left] node[left] {$a,b,c$} (011)
          (111) edge[loop left] node[left] {$a,b,c$} (111)
          (110) edge[loop right] node[right] {$a,b,c$} (110)
          (101) edge[loop below] node[below] {$a,b,c$} (101)
          (011') edge[loop left] node[left] {$a,b,c$} (011')
          (111') edge[loop left] node[left] {$a,b,c$} (111')
          (110') edge[loop right] node[right] {$a,b,c$} (110')
          (101') edge[loop below] node[below] {$a,b,c$} (101')
          (A) edge[loop above] node[above] {$a$} (A)
          (B) edge[loop above] node[above] {$b$} (B)
          (C) edge[loop above] node[above] {$c$} (C)
          (011'') edge[loop left] node[left] {$a,b,c$} (011'')
          (111'') edge[loop left] node[left] {$a,b,c$} (111'')
          (110'') edge[loop right] node[right] {$a,b,c$} (110'')
          (101'') edge[loop below] node[below] {$a,b,c$} (101'')
    
          ;

\begin{pgfonlayer}{background}

\filldraw [line width=14mm,line join=round,black!10]
      (A.north  -| A.east)  rectangle (A.south  -| A.west);

\filldraw [line width=14mm,line join=round,black!10]
      (B.north  -| B.east)  rectangle (B.south  -| B.west);

\filldraw [line width=14mm,line join=round,black!10]
      (C.north  -| C.east)  rectangle (C.south  -| C.west);

\filldraw [line width=12mm,line join=round,black!10]
      (011''.north  -| 101''.east)  rectangle (110''.south  -| 110''.west);
      
\filldraw [line width=12mm,line join=round,black!10]
      (011.north  -| 101.east)  rectangle (110.south  -| 110.west);

\filldraw [line width=12mm,line join=round,black!10]
      (011'.north  -| 101'.east)  rectangle (110'.south  -| 110'.west);

\filldraw [line width=12mm,line join=round,black!10]
      (@.north  -| @.east)  rectangle (@.south  -| @.west);
\end{pgfonlayer}

\end{tikzpicture}
}
    \caption{Pointed Kripke model representing the intermediate state of~\cref{MCEX:4} after all agents independently performed  a~priori belief updates in parallel but before any public announcements being made.}
    \label{fig:EX3.1}
\end{subfigure}

\par\bigskip 

\begin{subfigure}{.99\textwidth}
 \resizebox{10cm}{!}{
    \begin{tikzpicture}
\node[nod, label=135:{\textbf{\large{$V^a\upd\ASF$}}}] at (7,0) [label=below:$A$] (A) {};

\node (110) [right=6cm of A] {\textbf{\Large{$W^{\neg a}\upd\ASF$}}};

\node[nod, fill=gray, label=135:{\textbf{\Large{$W \upd \ASF$}}}] at (0,-3) [label=below:$\underline{0}$] (@') {};

\node at (-2.5,-1) [] {\textbf{\large{$(\calM \circledcirc_{a,b,c} \U'',\underline{0})\upd\ASF$}}};

\node[nod, label=135:{\textbf{\large{$V^b\upd\ASF$}}}] at (7,-3) [label=below:$B$] (A') {};

\node (110') [right=6cm of A'] {\textbf{\Large{$W^{\neg b}\upd\ASF$}}};

\node[nod, label=180:{\textbf{\large{$V^c\upd\ASF$}}}] at (7,-6) [label=below:$C$] (A'') {};
\node (110'') [right=6cm of A''] {\textbf{\Large{$W^{\neg c}\upd\ASF$}}};

\path[->, -{Stealth[scale width=1.5]}] 
(@')
edge node[below] {$a$} (A) 
(A) edge[loop above] node[above] {$a$} (A)
(@')
edge node[below] {$b$} (A')
(A') edge[loop above] node[above] {$b$} (A')
(@')
edge node[below] {$c$} (A'')
(A'') edge[loop above] node[above] {$c$} (A'');

\begin{pgfonlayer}{background}
\filldraw [line width=15mm,line join=round,black!10]
      (A.north  -| A.east)  rectangle (A.south  -| A.west);
\filldraw [line width=10mm,line join=round,black!10]
      (110.north  -| 110.east)  rectangle (110.south  -| 110.west);

\filldraw [line width=15mm,line join=round,black!10]
      (A'.north  -| A'.east)  rectangle (A'.south  -| A'.west);
\filldraw [line width=10mm,line join=round,black!10]
      (110'.north  -| 110'.east)  rectangle (110'.south  -| 110'.west);
\filldraw [line width=15mm,line join=round,black!10]
      (@'.north  -| @'.east)  rectangle (@'.south  -| @'.west);

\filldraw [line width=15mm,line join=round,black!10]
      (A''.north  -| A''.east)  rectangle (A''.south  -| A''.west);
\filldraw [line width=10mm,line join=round,black!10]
      (110''.north  -| 110''.east)  rectangle (110''.south  -| 110''.west);
\end{pgfonlayer}
\end{tikzpicture}
}
    \caption{Pointed Kripke model representing the final state of~\cref{MCEX:4} after all agents independently performed  a~priori belief updates in parallel and have applied the public announcement of~$\ASF$.}
    \label{fig:EX3.2}
\end{subfigure}

\caption{}
\end{figure}

In other words, after this a~priori belief update each child now erroneously believes that (a) it is the only muddy child and (b) other children ``erroneously'' believe to also be muddy.\footnote{The word \emph{erroneously} is here in air quotes because it has a flavor of a double Gettier example~\cite{Gettier63}. Firstly, $a$ believes that $b$~considers itself to be muddy, $b$ does in fact consider itself to be muddy, but not for the reason $a$ expects, as manifested by the difference between $b$'s part of the model representing $b$'s beliefs and $W^{\neg a}$ representing $a$'s rendition of $b$'s beliefs. Secondly, $a$ thinks that $b$ is wrong, i.e., that $b$ is clean, $a$ is in fact correct, but not for the reason $a$ expects, as manifested by the difference between the real world~$\underline{0}$ and world~$A$ of $V^a$.} Hence, all children step forward, a public update of $\ASF$ resulting in a model  shown in \cref{fig:EX3.2}.
This behavior conforms to everyone's expectations, so all children preserve consistency of beliefs. However, each child thinks that this behavior should have puzzled the other two, so that after stepping forward they think each other to have inconsistent beliefs:
\[
(\calM \circledcirc_{a,b,c} \U'')\upd \ASF,\underline{0} \models B_a ( B_b \bot \land B_c \bot)  \land B_b (B_a \bot \land B_c \bot) \land B_c (B_a \bot \land B_b \bot).
\]    
\end{example}

\subsection{Some Heuristics for A~priori Belief Updates}
\label{subsec:heuristics}
In developing a method for agents to ``think outside the box'' we did not want to put any boundaries or restrictions on the kinds of trial and backup models used in a~priori belief updates, did not want to ``box~them~in'' as it were. At the same time, there exist rather regular methods of generating a~priori belief updates, some of which we have already used. Let us outline some of these methods, which can easily be implemented in a form of an exhaustive trial-and-error search through finitely many possibilities.

\begin{paragraph}{Heuristics~1: Predefined master models}
    If an agent is aware of several incompatible alternative models of reality but is not sure which one of them better suits its observations and/or is not sure which points of view other agents entertain, then a~priori belief updates can be generated by assigning various combinations of these models to the agent itself and to other agents. This method was used in \cref{ex:consec2} for two different models of natural numbers.
\end{paragraph}

\begin{paragraph}{Heuristics~2: Varying explicit a~priori assumptions}
    Even when all agents agree on the same model, e.g., a master model representing the general rules of the puzzle, they may differ in additional a~priori assumptions each of them makes. Here the master model represents (common) implicit a~priori assumptions while additional assumptions are private and explicit in that they are represented by formulas, imposed on the master model. For instance, $\APB_2$ was such an explicit assumption in \cref{MCEX:2} imposed on the implicit assumptions modeled by $\calM_0$ from \cref{MCEX:1}. In many of our MCP-related examples, agent $a$  continues using \mbox{$\calM_0 \upd \APB_2$} as $\calM^{\neg a}$ to describe reasoning of others while herself switching to $\calM_0 \upd \APB_1$ as $\calM^a$ for her own reasoning, where $\APB_1 \ce m_a \lor m_b \lor m_c$ states that at least one child is muddy. If such an a~priori belief update fails to restore her consistency, she can then switch to the full $\calM_0$  or consider a more complex a~priori restriction.
\end{paragraph}

\begin{paragraph}{Heuristics~2a: Choosing explicit a~priori assumptions based on prior failures}
Using the master model, if the initial model or one of the previous unsuccessful attempts to reach a consistent state were based on some explicit a~priori beliefs $\APB$, then it is reasonable to attempt new explicit beliefs of the form $\APB' \land \neg \APB$ in constructing $\calM^a$ because the possibility of $\APB$ holding has already been rejected.
\end{paragraph}

\begin{paragraph}{Heuristics~3: Learning from inconsistencies arising from public announcements}
If an agent $a$ reaches an inconsistency as a result of the factual public announcement of $\varphi$, then $\varphi$ is a valuable information for guiding the a~priori reasoning process: in fact, to recover consistency, $\varphi$ should be true in at least one world of the equivalence class~$U^a$ of the trial model, so that when the public announcement is re-applied, the inconsistency won't arise again.
\end{paragraph}

\begin{paragraph}{Heuristics~4: Changing the underlying logic}
    It is also possible to relax the restrictions imposed on the trial and/or backup models. For instance, $a$ may try to explain the beliefs of $b$ and $c$ by the lack of positive and/or negative introspection on their part, which would necessitate the  relaxation of the requirements of transitivity and/or euclideanity on the backup model. 
\end{paragraph}

\section{Properties of A~priori Belief Updates}
\label{sec:propertiesAPBU}
This would be an appropriate place to list the axioms of a~priori belief updates, perhaps, akin to axioms of PAL for public announcements from~\cite{Plaza}. Unfortunately, this does not seem to be any easier to do than to provide a finite syntactic description of common epistemic puzzles. Even restricting our attention to finite models only, the question amounts to asking whether any finite combination of maximal consistent sets for the logic can be represented by a formula (or  finitely many formulas). The problem is that there are uncountably many maximal consistent sets (for an infinite set of  atomic propositions), hence, uncountably many finite combinations thereof, while only countably many formulas (cf.~\cite{Artemov22}). It seems difficult if not hopeless to try describing syntactically the result of an update  when no syntactic description of the update exists.\looseness=-1

Thus, in this section we try to ``model check'' rather than axiomatize the results of a~priori belief updates.

\begin{theorem}
\label{th:final_prop}
    Let\/ $\U = (\calM^a, U^a, \calM^{\neg a}, \C)$ be a single-agent a~priori belief update with trial model $\calM^a = \langle W^a, R^a, V^a\rangle$ and backup model $\calM^{\neg a} = \langle W^{\neg a}, R^{\neg a}, V^{\neg a}\rangle$ and let\/ $(\calM, v)$ be a pointed Kripke model  with $\calM = \langle W, R, V\rangle$  and $R_a(v) = \varnothing$. The following properties hold after agent~$a$'s~a~priori belief update:
    \begin{enumerate}
        \item For any formula $\psi$ that does not involve any modalities $B_b$ for agents~$b\ne a$, including for all purely propositional formulas, the following three statements are equivalent:
        \begin{enumerate}
        \item $\calM \circledcirc_a \U, v \models B_a \psi$;
        \item $\calM^a, u \models \psi$ for all $u \in U^a$;
        \item $\calM^a, u \models B_a \psi$ for at least one $u \in U^a$.
        \end{enumerate}
        In other words, agent~$a$'s factual beliefs and $a$'s beliefs about its own beliefs are fully determined by worlds from~$U^a$ of $\calM^a$.
        \item $\calM^{\neg a}, w \models \phi$ if{f} $\calM \circledcirc_a \U, w \models \phi$ for any  $w \in W^{\neg a}$ and any formula $\phi$.
        In other words, agent~$a$'s~higher-order beliefs about other agents' beliefs are fully determined by~$\calM^{\neg a}$.
    \end{enumerate}
\end{theorem}
\begin{proof}
    \begin{enumerate}
        \item The first statement easily follows from the fact that the identity relation on~$U^a$ is an $a$-bi\-simul\-at\-ion between $\calM^a$ and $\calM \circledcirc_a \U$ (where $a$-bisimulation means that forth and back conditions are restricted to $R_a$ transitions).
        \item This follows from the fact that the identity relation on $W^{\neg a}$ is a full bisimulation between Kripke models~$\calM^{\neg a}$ and~$\calM \circledcirc_a \U$.\qedhere
    \end{enumerate}
\end{proof}

\begin{remark}
    For an a~priori update $\U=(\calM^a,U^a,\calM^{\neg a}, \C)$, one
    might think that global properties of models~$\calM^a$~and~$\calM^{\neg a}$ would be transferred to $a$'s part of the model after the a~priori update~$\U$, so that after the update $a$ believes all global properties used in the update construction. This does hold for propositional validities, as follows from  the preceding theorem.  In fact, truth of $\psi$ in $U^a$ of $\calM^{a}$ is already sufficient to ensure $a$'s post-update belief $B_a \psi$ if $\psi$ is purely propositional. However, this  transfer of global properties fails for epistemic formulas involving other agents. Consider, for instance, the a~priori update $\U$ from~\cref{fig:EX1.1} for \cref{MCEX:2}. For formula $\ASF$ from~\eqref{eq:asf}, it is clear that $\calM^a\models \neg \ASF$ and $\calM^{\neg a} \models \neg \ASF$. After all, we know that, in the standard MCP, children do not step forward all at once unless all are muddy. However, in model $\calM \circledcirc_a \U$ from \cref{fig:EX1.3}, we have $\calM \circledcirc_a \U, A \models \ASF$, resulting in $\calM \circledcirc_a \U, \underline{A} \models B_a\ASF$. Thus, global assumptions about other agents' beliefs need not survive in the face of the mismatch between two different points of view: the one $a$~reserves for itself vs.~the one $a$~assigns to others. This negative result also extends to implicit assumptions such as factivity of beliefs since in the same example both $\calM^a$~and~$\calM^{\neg a}$~represented agents with factive beliefs, but the model after the a~priori update lacks reflexivity for $b$ and $c$ predictably resulting in their beliefs not being factive.\looseness=-1
\end{remark}

 \begin{theorem}
 \label{theo:no_info}
 When agent~$a$'s a~priori update\/ $\U$ triggered by a public statement $\phi$ of agent~$b$ results in agent~$a$ believing that agent~$b$ has inconsistent beliefs, (re)applying $b$'s public statement after the a~priori update does not affect $a$'s~factual beliefs, as well as $a$'s beliefs about beliefs of $a$ and $b$. In other words, if $\calM\upd B_b \phi, v \models B_a \bot$ and\/ $(\calM\upd B_b \phi) \circledcirc_a \U, v \models B_a B_b \bot$, then for any formula $\psi$ that does not involve any modalities $B_c$ for agents~$c \notin\{a,b\}$,  
 \[
 (\calM\upd B_b \phi) \circledcirc_a \U, v \models B_a \psi 
 \qquad \Longleftrightarrow\qquad 
 \left((\calM\upd B_b \phi) \circledcirc_a \U\right) \updpriv{a} B_b \phi, v \models B_a \psi.
\]
 \end{theorem}

 \begin{proof}
 Indeed, $(\calM \upd B_b \phi)\circledcirc_a \U, v \models B_a B_b \bot$ means that there are no $b$-outgoing arrows from any world of $U^a$ in $(\calM \upd B_b \phi)\circledcirc_a \U$. While the public update may affect worlds in $W^{\neg a}$ and, thereby, what $a$ thinks of beliefs of other agents (not $a$ or $b$), the $U^a$ part of the model before the public announcement update is $ab$-bisimilar to itself after the public announcement.
 \end{proof}

\section{Conclusions}
\label{sec:concl}

To the best of our knowledge, this paper provides the first epistemic formalization of 
a~priori belief updates, which are crucial for  modeling of self-adaptive systems, where agents should possess self-correcting abilities. 

We provide multiple examples demonstrating how such self correction can be achieved in various versions of standard epistemic puzzles, e.g.,~consecutive numbers and muddy children, in cases where agents find themselves in an inconsistent epistemic state. We also provide examples of self correction not resolving the inconsistency or resolving it in a way that results in several or even all agents arriving at false conclusions. In one case, the false conclusions include the inconsistency of other agents' beliefs, which prevents any future corrections no matter which communications will occur.  We prove properties of a~priori belief updates semantically, explain the difficulties of developing a logic of a~priori updates,  and provide counterexamples to several preservation properties, underscoring the arbitrary nature of a~priori updates.

\paragraph{Related work.} We believe that some aspects of the presented a~priori belief updates could be implemented in plausibility models via plausibility change~\cite{vDitmarsch05,BaltagS08,vBenthemS15}. However, this would require to embed all a~priori updates an agent might ever consider into the initial model. Not only would this lead to a significant increase of the initial model size, but it would also mean that all possible self-correcting subroutines, as well as their order of application, are pre-programmed, turning them from autonomous a priori actions into deterministic a posteriori private updates. This approach may have merits for traditional distributed systems but is less suitable for SASO systems.

\paragraph{Future work.} We aim at extending the update mechanism to more general action models~\cite{ditmarsch2007dynamic} and developing  logics for restricted types of a~priori belief updates, which may find applications to belief update synthesis problem.
We also plan to address the question of communicating new a~priori beliefs (i.e.,~modeling a~posteriori updates about a~priori beliefs), which would enable other agents to guide the recovery (or dually, state corruption) of the agent in question.

It would also be interesting to explore how  these methods would look in the formalism of simplicial complexes~\cite{Goubault_2018,vDitmarschGLLR21}, which are dual to Kripke models. For instance, the operation of choosing a new local state, which amounts to identifying a suitable equivalence class in a Kripke model, would, in a simplicial model, correspond to choosing a new singleton vertex. This appears to be a more basic operation suggesting that the process of a~priori belief updates may look more natural when presented  simplicially.

\paragraph{Acknowledgments.}  We are grateful to Hans van Ditmarsch for his discerning comments on an earlier version of this paper, as well as to Stephan Felber, Kristina Fruzsa, Rojo Randrianomentsoa, Hugo Rinc\'on Galeana, Ulrich Schmid, and Thomas Schl\"ogl 
for multiple illuminating and inspiring discussions.

\bibliographystyle{alphaurl}
\bibliography{references1}

\end{document}